\documentclass[letterpaper, 10 pt, conference]{ieeeconf}
\IEEEoverridecommandlockouts
\usepackage{graphicx}
\usepackage{amsmath,amssymb,bm}
\let\labelindent\relax
\usepackage{enumitem}
\usepackage{xcolor}
\usepackage{hyperref}

\newcommand{\real}{\ensuremath{\mathbb{R}}}
\newcommand{\smat}[1]{\ensuremath{\left[\begin{smallmatrix}#1\end{smallmatrix}\right]}}
\newcommand{\bmat}[1]{\ensuremath{\begin{bmatrix}#1\end{bmatrix}}}

\makeatletter
\def\ifemptyarg#1{%
  \if\relax\detokenize{#1}\relax 
    \expandafter\@firstoftwo
  \else
    \expandafter\@secondoftwo
  \fi}
\makeatother
\newcommand{\mb}[1]{\mathbf{#1}}
\newcommand{\Zc}{\boldsymbol{\zeta}}
\newcommand{\AB}[1][]{\ifemptyarg{#1}{\ensuremath{\bmat{A~\,B}}}{\ensuremath{\bmat{A_{#1}~\,B_{#1}}}}}
\newcommand{\ABC}[1][]{\ifemptyarg{#1}{\ensuremath{\big[A~\,B~\,C\big]}}{\ensuremath{\big[A_{#1}~\,B_{#1}~\,C_{#1}\big]}}}
\newcommand{\B}[1][]{\ifemptyarg{#1}{\ensuremath{\mathcal{B}}}{\ensuremath{\mathcal{B}_{#1}}}}
\newcommand{\ABCd}[1][]{\ifemptyarg{#1}{\ensuremath{\big[A~\,B~\,C~\,d\big]}}{\ensuremath{\big[A_{#1}~\,B_{#1}~\,C_{#1}~\,d_{#1}\big]}}}

\DeclareMathSymbol{\cdoT}{\mathord}{symbols}{"01}
\DeclareMathOperator{\diag}{diag}
\DeclareMathOperator{\Tr}{Tp}

\newtheorem{remark}{Remark}
\newtheorem{lemma}{Lemma}
\newenvironment{proofof}[1]{\textit{Proof of #1.}}{\hfill \hspace*{0.01pt}  \hfill$\square$}
\newtheorem{problem}{Problem}
\newtheorem{theorem}{Theorem}
\newtheorem{assumption}{Assumption}

\allowdisplaybreaks

\title{\LARGE \bf
Setpoint control of bilinear systems from noisy data
}

\author{Andrea Bisoffi$^{1}$, Dominiek M. Steeman$^{2}$ and Claudio De Persis$^{2}$
\thanks{This publication is part of the project Digital Twin with project number
P18-03 of the research programme TTW Perspective which is (partly)
financed by the Dutch Research Council (NWO).}
\thanks{$^{1}$A. Bisoffi is with the Department of Electronics, Information,
and Bioengineering, Politecnico di Milano, 20133, Italy
        {\tt\small\{andrea.bisoffi\}@polimi.it}}%
\thanks{$^{2}$D. M. Steeman, C. De Persis are with the Engineering and Technology
Institute, University of Groningen, 9747 AG, The Netherlands
        {\tt\small \{d.m.steeman@student., c.de.persis@\}rug.nl}}%
}

\begin{document}

\maketitle

\begin{abstract}
We consider the problem of designing a controller for an unknown bilinear system using only noisy input-states data points generated by it.
The controller should achieve regulation to a given state setpoint and provide a guaranteed basin of attraction.
Determining the equilibrium input to achieve that setpoint is not trivial in a data-based setting and we propose the design of a controller in two scenarios. The design takes the form of linear matrix inequalities and is validated numerically for a \'Cuk converter.
\end{abstract}

\section{Introduction}
\label{sec:intro}

A popular modus operandi in the recent literature on data-driven control design has been the following.
From a dynamical system, assumed to be in a certain class, data affected by some kind of noise are collected.
Only a bound on the noise is assumed to be known, and the uncertainty associated with the noise translates into a set of possible dynamical systems of the given structure which can explain data and among which the system generating data cannot be distinguished.
This calls for designing a controller that robustly stabilizes all dynamical systems in such set.
This modus operandi, which combines set membership identification \cite{milanese2004set} and robust control \cite{schererweiland} in a seamless algorithm for control design from data, has been pursued in \cite{coulson2019data,depersis2019formulas,vanwaarde2020noisy,berberich2022combining}, to name but a few recent works.
Here, we follow this approach for the class of bilinear dynamical systems.
Relevantly, bilinear systems are a halfway house between linear systems, which are endowed with strong structure and properties, and nonlinear systems, which are less tractable but capture more complex phenomena.
Also, approaches such as Carleman linearization \cite[p.~110]{rugh1981nonlinear} or Koopman operator \cite{goswami2022bilinearization} lead to bilinear systems after truncation and, in this sense, bilinear system approximate more complex nonlinear ones.

Our results have, as a starting point, the noisy data collected from a bilinear system.
Using only these data and the noise bound, but \emph{without} knowledge of the bilinear system parameters, our goal is to asymptotically stabilize a given state setpoint, in spite of the uncertainty associated with data, while providing a guaranteed basin of attraction for it.
In a model-based setting, stabilization of a nonzero setpoint entails finding, in the first place, an equilibrium input that enforces the nonzero setpoint;
when one is found, stabilization of this nonzero setpoint boils down to stabilization of a zero setpoint (the origin), by a change of coordinates.
We show here that, in a data-based setting, requiring stabilization of a nonzero setpoint has some interesting implications.

\emph{Related literature.}  We considered data-driven control of bilinear systems for a zero setpoint in our previous work \cite{bisoffi2020bilinear} and here we improve on \cite{bisoffi2020bilinear} by considering a wider setting (multiple inputs, discrete and continuous time) and less conservative assumptions (no knowledge of upper bounds on the norm of system matrices).
A detailed comparison is in Remark~\ref{remark:comparison}.
A data-driven control result for bilinear systems is given in~\cite{yuan2022data} but, unlike here, noise-free data are considered.
In a model-based setting, stabilization of bilinear systems via linear state feedback and a quadratic or polyhedral Lyapunov function is considered in \cite{tarbouriech2009control,bitsoris2008ifac,khlebnikov2018quadratic,amato2009stabilization} 
and such Lyapunov function can be leveraged to give an underapproximation of the basin of attraction \cite{kramer2021stability}.
We proceed in an analogous way, but our design is fully data-driven and, unlike these works, does not require the knowledge of the matrices of the bilinear system.
Incidentally, this has some interesting consequences for nonzero setpoints.
Aside from control design, \cite{favoreel1999subspace,chen2000subspace,sontag2009input} tailor system identification methods for bilinear systems.
More recently, \cite{markovsky2022data} pursues simulation of trajectories of so-called generalized bilinear systems from data by embedding such a generalized bilinear system into a linear one; \cite{karachalios2022framework} constructs a quadratic-bilinear system that fits input-output data; \cite{sattar2022finite} proposes a finite-sample analysis for the error rate in learning bilinear systems.

\emph{Contribution.} 
For setpoint control of bilinear systems from data, we propose two designs, one for when the equilibrium input corresponding to the given setpoint is known and one for when it is not.
The first one improves on \cite{bisoffi2020bilinear} by, most notably, not requiring any knowledge of upper bounds on the norms of system matrices, see Remark~\ref{remark:comparison}.
The second one highlights some interesting peculiarities of setpoint control in a noisy data-based setting.

\emph{Structure.} Section~\ref{sec:data} formulates the problem. 
Sections~\ref{sec:known_u_bar} and \ref{sec:UNknown_u_bar} propose two solutions to the problem, for two different scenarios.
Section~\ref{sec:example} validates the designs on a numerical example.

\emph{Notation.}
The Kronecker product between two matrices is denoted by $\otimes$ and its main properties are in~\cite[\S 4.2.]{horn1994topics}.
For natural numbers $n$ and $m$, $I_n$ (or $I$) is an identity matrix of dimension $n$ (or of suitable dimension) and $0_{m,n}$ (or $0$) is a matrix of zeros of dimension $m$-by-$n$ (or of suitable dimension).
The induced 2-norm of a matrix $A$ is $\| A\|$; for a scalar $a \ge 0$, $\| A \| \le a$ if and only if $A^\top A \preceq a^2 I$.
For a matrix $A$ with real entries, the transposition operator is $\Tr A := A + A^\top$.
For matrices $A$, $B$ and $C$ of compatible dimensions, we abbreviate $A B  C (AB)^\top$ to $A B \cdoT C[\star]^\top$ or $[\star]^\top C \cdot (AB)^\top$, where the dots in the expressions clarify unambiguously which terms are to be transposed.
For matrices $A=A^\top$, $B$, $C=C^\top$, we abbreviate the symmetric matrix $\smat{A & B \\ B^\top & C}$ as $\smat{A & B \\ \star & C}$ or $\smat{A & \star \\ B^\top & C}$.

\section{Data collection, problem statement, preliminary reformulations}
\label{sec:data}

Consider a bilinear system
\begin{equation}
\label{sys}
x^\circ = A_\star x + B_\star u + C_\star (u \otimes x) + d_\star
\end{equation}
where $x \in \real^n$ is the state, $u \in \real^m$ is the input and $x^\circ$ is the update $x^+$ of $x$ in discrete time or the time derivative $\dot x$ of $x$ in continuous time.
\eqref{sys} is equivalently reformulated as
\begin{equation}
\label{sys_alt}
x^\circ = A_\star x + B_\star u + C_\star (I_m \otimes x) u + d_\star,
\end{equation}
since $u \otimes x = (I_m u) \otimes (x\cdot 1) = (I_m \otimes x) u$ for all $x$ and $u$.
The constant matrices $A_\star$, $B_\star$, $C_\star$, $d_\star$ in~\eqref{sys_alt} are \emph{unknown} to us and we rely instead on noisy data collected through an experiment on the system.
The data points collected at times $t_0$, $t_1$, \dots, $t_{T-1}$ satisfy
\begin{align}
& \hspace*{-3pt} x^\circ(t_i) = A_\star x(t_i) + B_\star u(t_i) \notag \\
& \hspace*{-3pt} + C_\star (u(t_i) \otimes x(t_i)) + d_\star + e(t_i), \,\,  i = 0, 1,\dots, T-1 \label{data_vec}
\end{align}
where $e(t_0)$, $e(t_1)$, \dots, $e(t_{T-1})$ is the sequence of unknown process noise perturbing the data points of the experiment.
The relations in~\eqref{data_vec} can be written compactly as
\begin{equation}
X_1 = A_\star X_0 + B_\star U_0 + C_\star S_0 + d_\star O_0 + E_0 \label{data_matr_eq}
\end{equation}
by introducing the measured sequences
\begin{subequations}
\label{data_mat}
\begin{align}
X_1 & := \bmat{x^\circ(t_0)  & \dots  & x^\circ(t_{T-1})} \\
& =
\begin{cases}
\bmat{x(t_0 +1)  & \dots  & x(t_{T-1}+1)} & \text{ in discr. time}\\
\bmat{\dot x(t_0)  & \dots  & \dot x(t_{T-1})} & \text{ in cont. time}
\end{cases} \notag \\
X_0 & := \bmat{x(t_0)  & \dots  & x(t_{T-1})}\\
U_0 & := \bmat{u(t_0)  & \dots & u(t_{T-1})}\\
S_0 & := \bmat{ u(t_0)\otimes x(t_0) & \dots &  u(t_{T-1}) \otimes x(t_{T-1})},\label{data-bil-S0}
\end{align}
\end{subequations}
an auxiliary sequence and the unknown noise sequence
\begin{align*}
O_0 := \bmat{ 1 & \dots & 1}
\text{ and }
E_0 := \bmat{e(t_0)  & \dots & e(t_{T-1})}.
\end{align*}
In continuous time, we assume for simplicity to measure the state derivative.
When not available, $\dot{x}$ can be recovered, e.g., from dense sampling of the state with a reconstruction error captured by $e$; see more details in~\cite[\S 3.1 and Remark 1]{luppi2024data}.

When data points are collected from a physical plant, it is reasonable to assume that the noise sequence is bounded in some sense.
Here, we consider an energy bound on the noise sequence of the experiment, i.e., we assume that $E_0$ belongs to
\begin{equation}
\label{set_E}
\mathcal{E} := \{ E \in \real^{n \times T} \colon E E^\top \preceq \Xi \Xi^\top \}
\end{equation}
where $\Xi \Xi^\top$ is some positive semidefinite matrix (by construction).
We consider an energy bound for simplicity but different types of bounds can be handled with the techniques in~\cite{bisoffi2022petersen}.

As in~\cite{milanese2004set}, \eqref{data_matr_eq} and \eqref{set_E} lead to a set of \emph{dynamics consistent with data}, that is, a set of matrices $\ABCd$ that could have generated $X_1$, $X_0$, $U_0$, $S_0$ for a noise sequence $E \in \mathcal{E}$, that is,
\begin{align}
\label{set_C}
& \mathcal{C} := \{ \ABCd \colon \notag \\
& \quad X_1 = A X_0 + B U_0 + C S_0 + d O_0 + E, E \in \mathcal{E} \}.
\end{align}
Note that $E_0 \in \mathcal{E}$ if and only if $\ABCd[\star] \in \mathcal{C}$.
With
\begin{equation}
\label{matrix W_0}
W_0 := \smat{ X_0\\ U_0\\S_0\\ O_0},
\end{equation}
we make the next assumption on data.
\begin{assumption}
\label{ass:pers ext-bil}
Matrix $W_0$ in~\eqref{matrix W_0} has full row rank.
\end{assumption}
Full row rank of $W_0$ is related to persistence of excitation of the input and noise sequences, see the discussion in \cite[\S 4.1]{bisoffi2022petersen}.
Assumption~\ref{ass:pers ext-bil} is verified directly from data and, when $W_0$ does not have full row rank, it can typically be enforced by collecting more data points, thus adding columns to $W_0$.

For this setting, we can state our objective next.
\begin{problem}
\label{problem:statement}
Using only the noisy data in~\eqref{data_mat}, which satisfy Assumption~\ref{ass:pers ext-bil}, and given a desired setpoint $\bar{x}$ for the state, design parameters $K$ and $\bar u$ of an affine control law
\begin{align}
\label{control_law}
u = K (x - \bar{x}) + \bar{u}
\end{align}
so that, for the feedback interconnection of \eqref{sys_alt} with~\eqref{control_law}, either $\bar{x}$ or a ``small'' neighborhood of $\bar{x}$ is locally asymptotically stable (as a set), while providing a guaranteed basin of attraction.
\end{problem}

We now discuss Problem~\ref{problem:statement} and outline how we intend to solve it.
The selection of an affine control law is common in the model-based literature \cite{tarbouriech2009control,khlebnikov2018quadratic,bitsoris2008ifac,amato2009stabilization} since it enables pursuing the design of $K$ through linear matrix inequalities.
In a model-based setting, given the desired setpoint $\bar{x}$, one would select $\bar{u}$ based on the ``true'' $\ABCd[\star]$ so that $(\bar{x},\bar{u})$ is an equilibrium for the open loop in~\eqref{sys_alt}, i.e., given $\bar{x}$ and $\ABCd[\star]$, one would find $\bar{u}$ to satisfy
\begin{align}
\label{equilibrium_eq}
& \left. \begin{matrix}
\hspace*{-5pt}\text{in discr. time} & \hspace*{-5pt}\bar{x}\\
\hspace*{-5pt}\text{in cont. time} & \hspace*{-5pt}0
\end{matrix}\right\} \!\! = \! A_\star \bar{x} + B_\star \bar{u} + C_\star (I_m \otimes \bar{x}) \bar u + d_\star.
\end{align}
In a data-based setting, the ``true'' $\ABCd[\star]$ are not available, but one could still perform a(n additional) static experiment on the plant to find a $\bar{u}$ such that the constant (measured) solution $x(\cdot) = \bar{x}$ satisfies \eqref{equilibrium_eq}.
Such static experiments that aim at characterizing the \emph{static} relation between $\bar{u}$ and $\bar{x}$ are common and can be relatively cheap.
If such an additional experiment can be performed, we can assume that we know a $\bar{u}$ such that $(\bar{x},\bar{u})$ is an equilibrium point for~\eqref{sys_alt}.
Another case when we know such $\bar{u}$ is when $d_\star =  \bar{x} = 0$, which is typically considered in the literature of control design for bilinear system.
Indeed, for $d_\star = \bar{x} = 0$, a $\bar{u}$ satisfying \eqref{equilibrium_eq} for any unknown $\ABC[\star]$ is $\bar{u} = 0$.
For these reasons, it is relevant to first address the case of knowing $\bar{u}$ such that $(\bar{x},\bar{u})$ is an equilibrium for~\eqref{sys_alt} in Section~\ref{sec:known_u_bar}.
We then address the case when $\bar{u}$ is \emph{unknown} and thus treated as a design parameter in Section~\ref{sec:UNknown_u_bar}.
This will entail that we cannot make $\bar{x}$ locally asymptotically stable, but only a ``small'' neighborhood of it.
Incidentally, the solution in Section~\ref{sec:known_u_bar} improves upon the results in \cite{bisoffi2020bilinear}, as we discuss below in Remark~\ref{remark:comparison}.

We conclude this section by reformulating dynamics and set of dynamics consistent with data.
For all $\bar{x}$ and $\bar{u}$ in~\eqref{control_law}, the feedback interconnection of \eqref{sys_alt} and \eqref{control_law} yields
\begin{align}
& x^\circ = A_\star x + B_\star \big( K (x - \bar{x}) + \bar{u} \big) \notag \\
& + C_\star (I_m \otimes x) \big( K (x - \bar{x}) + \bar{u} \big) + d_\star.
\label{sys+control_law}
\end{align}
For later analysis, it is useful to consider the change of coordinates $\tilde{x} := x - \bar{x}$ and rewrite \eqref{sys+control_law} equivalently as
\begin{align}
& \left. \begin{matrix}
\text{in discr. time} & \tilde{x}^+ + \bar{x}\\
\text{in cont. time} & \dot{\tilde{x}}
\end{matrix}\right\} = \notag \\
& = \ABC[\star] \smat{I\\ K\\ (I_m \otimes \bar{x})K + \bar{u} \otimes I_n + (I_m \otimes \tilde{x}) K} \tilde{x} \notag \\
& + \ABCd[\star]\smat{\bar{x}\\ \bar{u}\\ (I_m \otimes \bar{x}) \bar{u}\\ 1}
\label{sys_tilde}
\end{align}
since $(I_m \otimes \tilde x) \bar{u}  = \bar{u} \otimes \tilde x = (\bar{u} \cdot 1) \otimes (I_n \tilde x) = ( \bar{u} \otimes I_n) \tilde x$ for all $\tilde{x}$ and $\bar{u}$.

The set $\mathcal{C}$, which is key in the sequel, can be reformulated in the same way as in \cite[\S 2.3]{bisoffi2022petersen}.
By algebraic computations,
\begin{subequations}
\label{set C form 1}
\begin{align}
& \mathcal{C} = \{ \ABCd = Z^\top \colon \bmat{I_n & Z^\top} \bmat{\mb{C} & \mb{B}^\top\\ \mb{B} & \mb{A} } \bmat{I_n\\ Z} \preceq 0 \} \label{set C form 1:set only}\\
& \mb{C} := X_1 X_1^\top-\Xi\Xi^\top, \, \mb{B} := - W_0 X_1^\top,\, \mb{A} := W_0 W_0^\top
\label{set C form 1:ABC}
\end{align}
\end{subequations}
and also, since $\mb{A} \succ 0$ by Assumption~\ref{ass:pers ext-bil}, to
\begin{subequations}
\label{set C form 2}
\begin{align}
& \!\!\!\mathcal{C} \!=\! \big\{ \ABCd = Z^\top \! \colon (Z-\Zc)^\top \mb{A} (Z - \Zc ) \preceq \mb{Q} \big\} \\
& \!\!\!\Zc := - \mb{A}^{-1} \mb{B},\, \mb{Q} := \mb{B}^\top \mb{A}^{-1} \mb{B} - \mb{C}. \label{set C form 2:Zc and Q}
\end{align}
\end{subequations}
A last form of $\mathcal{C}$ and its relevant properties are given next.
\begin{lemma}
Under Assumption~\ref{ass:pers ext-bil}, we have: $\mb{A} \succ 0$, $\mb{Q} \succeq 0$, $\mathcal{C}$ is bounded with respect to any matrix norm, and 
\begin{equation}
\label{set C form 3}
\mathcal{C} = \big\{ (\Zc+ \mb{A}^{-1/2} \Upsilon \mb{Q}^{1/2})^\top \colon \|\Upsilon\| \le 1 \big\}.
\end{equation}
\end{lemma}
\begin{proof}
Repeat the proofs of \cite[Lemma~1]{bisoffi2022petersen}, \cite[Lemma~2]{bisoffi2022petersen}, \cite[Prop.~1]{bisoffi2022petersen} after replacing $\smat{X_0 \\ U_0}$, $\Delta \Delta^\top$, $\AB[\star]$ with $W_0$, $\Xi \Xi^\top$, $\ABCd[\star]$, respectively.
\end{proof}

\section{Data-driven setpoint control with known equilibrium input}
\label{sec:known_u_bar}

In this section we address Problem~\ref{problem:statement} when we know a value $\bar{u}$ such that $(\bar{x},\bar{u})$ is an equilibrium for~\eqref{sys_alt}, i.e., $(\bar{x},\bar{u})$ satisfies \eqref{equilibrium_eq}.
In continuous time, \eqref{sys_tilde} simplifies, by~\eqref{equilibrium_eq}, to
\begin{align}
& \dot{\tilde x} = 
\ABC[\star] \smat{I\\ K\\ (I_m \otimes \bar{x}) K + \bar{u} \otimes I_n + (I_m \otimes \tilde{x}) K} \tilde{x}. \label{sys_ct_tilde_known_u_bar}
\end{align}
The matrix-ellipsoid parametrization in \eqref{set C form 3} along with the result known as Petersen's lemma, see \cite{bisoffi2022petersen}, allows us to obtain the next result for data-driven control of the bilinear system in~\eqref{sys_alt}, ensuring local asymptotic stability of $\bar{x}$ with a guaranteed basin of attraction.
\begin{theorem}
\label{thm:known_u_bar:ct}
Under Assumption~\ref{ass:pers ext-bil}, let data $X_1$, $X_0$, $U_0$, $S_0$ in \eqref{data_mat} yield $\mb{A}$, $\Zc$, $\mb{Q}$ in \eqref{set C form 1:ABC}, \eqref{set C form 2:Zc and Q} and let a $\bar{u}$ 
satisfying~\eqref{equilibrium_eq} be known.
Suppose the next program is feasible
\begin{subequations}
\label{sol-bil:ct}
\begin{align}
& \text{find} & & Y,~P=P^\top \succ 0,~\lambda >0,~\Lambda >0 \label{sol-bil:find:ct}\\
& \text{s.t.}  & & 
\left[
\begin{matrix}
\Tr\left\{ \smat{P\\ Y \\ (I_m \otimes \bar{x}) Y + (\bar{u} \otimes I_n)P \\ 0}^\top  \Zc   \right\} & \star \\
\smat{0\\0 \\ I_m \otimes P\\ 0}^\top \Zc & - \lambda (I_m \otimes P) \\
\lambda Y & 0  \\
\mb{A}^{-1/2} \smat{P\\ Y \\ (I_m \otimes \bar{x}) Y + (\bar{u} \otimes I_n)P \\ 0} & \mb{A}^{-1/2} \smat{0\\0 \\ I_m \otimes P \\ 0} \\
\Lambda \mb{Q}^{1/2} & 0 
\end{matrix}
\right.\notag \\
& & & \left.
\begin{matrix}
\dots & \star & \star & \star \\
\dots & \star & \star & \star \\
\dots & - \lambda I & \star & \star\\
\dots & 0 & - \Lambda I & \star \\
\dots & 0 & 0 & - \Lambda I
\end{matrix}
\right]\prec 0 \label{sol-bil:qlmi:ct}
\end{align}
\end{subequations}
and let $P$ and $K = Y P^{-1}$ be a solution to it.
For
\begin{align*}
\dot x = A_\star x + B_\star u + C_\star (I_m \otimes x) u + d_\star, u = K(x - \bar{x}) + \bar{u},
\end{align*}
$\bar{x}$ is locally asymptotically stable with basin of attraction including $
\mathcal{B}_{P}^{\bar{x}} :=\{ x \in \real^n \colon (x-\bar{x})^\top P^{-1} (x-\bar{x}) \le 1 \}$.
\end{theorem}
\begin{proof}
Since $\bar{u}$ satisfies \eqref{equilibrium_eq} and by the change of coordinates $\tilde{x} := x - \bar{x}$, the statement becomes that for~\eqref{sys_ct_tilde_known_u_bar}, the origin is locally asymptotically stable with basin of attraction including
$\mathcal{B}_P:=\{ \tilde x \in \real^n \colon \tilde x^\top P^{-1} \tilde x \le 1 \}$.
To prove this, we need to show that, for $P=P^\top \succ 0$, the derivative of the Lyapunov function $\tilde{x}^\top P^{-1} \tilde{x}$ along solutions to~\eqref{sys_ct_tilde_known_u_bar} satisfies for all $\tilde{x} \in \mathcal{B}_P \backslash \{0\}$
\begin{align*}
& \!\Tr \{ \tilde{x}^\top P^{-1} \ABC[\star] \! \smat{I\\ K\\ (I_m \otimes \bar{x}) K + \bar{u} \otimes I_n + (I_m \otimes \tilde{x}) K} \! \tilde{x} \} \! < \! 0. 
\end{align*}
Since $\ABC[\star]$ is unknown to us, we show instead that for all $\ABCd \in \mathcal{C}$ and all $\tilde{x} \in \mathcal{B}_P \backslash \{ 0\}$,
\begin{align}
& \Tr \{ \tilde{x}^\top P^{-1} \ABC \smat{I\\ K\\ (I_m \otimes \bar{x}) K + \bar{u} \otimes I_n + (I_m \otimes \tilde{x}) K} \tilde{x} \} < 0.
\label{bil-decrease-Lyap:ct}
\end{align}
This holds if for all $\ABCd \in \mathcal{C}$, all $\tilde{x} \in \mathcal{B}_P \backslash \{0\}$ and all $\tilde{y} \in \mathcal{B}_P$,
\begin{align*}
& \Tr \{ \tilde{x}^\top P^{-1} \ABC \smat{I\\ K\\ (I_m \otimes \bar{x}) K + \bar{u} \otimes I_n + (I_m \otimes \tilde{y}) K} \tilde{x} \} < 0. 
\end{align*}
This holds if and only if for all $\ABCd \in \mathcal{C}$ and all $\tilde{y} \in \mathcal{B}_P$,
\begin{align*}
& \Tr \{ P^{-1} \ABC \smat{I\\ K\\ (I_m \otimes \bar{x}) K + \bar{u} \otimes I_n + (I_m \otimes \tilde{y}) K} \} \prec 0.
\end{align*}
By pre\slash post-multiplying by $P \succ 0$ and setting $K P = Y$, this holds if and only if for all $\ABCd \in \mathcal{C}$ and all $\tilde{y} \in \mathcal{B}_P$,
\begin{align*}
& 0 \succ \Tr \{ \ABC \smat{P\\ Y\\ (I_m \otimes \bar{x}) Y + (\bar{u} \otimes I_n)P + (I_m \otimes \tilde{y}) Y} \} \\
& = \Tr \{ \ABC \smat{P\\ Y\\ (I_m \otimes \bar{x}) Y + (\bar{u} \otimes I_n)P } \} \\
& +  \Tr \{ \ABC \smat{0\\ 0\\ I_{mn} } (I_m \otimes \tilde{y}) Y \}.
\end{align*}
By $P \succ 0$, $\tilde{y} \tilde{y}^\top \preceq P$ is the same as $\tilde{y} \in \mathcal{B}_P$ and $\tilde{y} \tilde{y}^\top \preceq P$ is equivalent to $(I_m \otimes \tilde{y})(I_m \otimes \tilde{y})^\top \preceq I_m \otimes P$.
Then, the previous condition holds if
\footnote{
We note why this is only an implication and not an equivalence.
The implication corresponds to the fact that $\mathcal{S}_1 := \{ V = I_m \otimes \tilde{y}^\top \colon \tilde{y} \in \real^n, \tilde{y} \tilde{y}^\top \preceq P \} \subseteq \{ V \in \real^{m \times mn } \colon V^\top V \preceq I_m \otimes P \} =: \mathcal{S}_2$ since $(I_m \otimes \tilde{y})(I_m \otimes \tilde{y})^\top = I_m \otimes (\tilde{y} \tilde{y}^\top) \preceq I_m \otimes P$ by $\tilde{y} \tilde{y}^\top \preceq P$.
On the other hand, there are counterexamples to $\mathcal{S}_1 \supseteq \mathcal{S}_2$ so that the equivalence does not hold.
Indeed, for $m=n=2$ and $P = \smat{p_{11} & p_{12}\\ p_{12} & p_{22}} \succ 0$, take $v = \frac{1}{2}(p_{11} - \frac{p_{12}^2}{p_{22}}) > 0$ and $V= \smat{0 & 0 & 0 & 0\\ \sqrt{v} & 0 & 0 & 0}$; $V$ satisfies $V^\top V  = \smat{v & 0 & 0 & 0\\ 0 & 0 & 0 & 0\\ 0 & 0 & 0 & 0\\ 0 & 0 & 0 & 0} \preceq I_2 \otimes P$ and thus $V \in \mathcal{S}_2$; however, $V \notin \mathcal{S}_1$ since it cannot be written, for $\tilde{y} = \smat{\tilde{y}_1\\ \tilde{y}_2}$, as $\smat{\tilde{y}_1 & \tilde{y}_2 & 0 & 0\\ 0 & 0 & \tilde{y}_1 & \tilde{y}_2}$.}
for all $\ABCd \in \mathcal{C}$ and for all $V$ with $V^\top V \preceq I_m  \otimes P$,
\begin{align*}
& 0 \succ \Tr \{ \ABC \smat{P\\ Y\\ (I_m \otimes \bar{x}) Y + (\bar{u} \otimes I_n)P } \} \\
& +  \Tr \{ \ABC \smat{0\\ 0\\ I_{mn} } V^\top Y \}.
\end{align*}
By Petersen's lemma as reported in \cite[Fact~1]{bisoffi2022petersen} and $P \succ 0$, this holds if and only if for all $\ABCd \in \mathcal{C}$, there exists $\lambda > 0$ such that
\begin{align*}
& 0 \succ \Tr \{ \ABC \smat{P\\ Y\\ (I_m \otimes \bar{x}) Y + (\bar{u} \otimes I_n)P } \}  + \lambda Y^\top Y \\
& +  \frac{1}{\lambda} \ABC \smat{0\\ 0\\ I_{mn} } (I_m \otimes P) \smat{0\\ 0\\ I_{mn} }^\top \ABC^\top.
\end{align*}
By Schur complement, $\lambda > 0$ and $P \succ 0$, this holds if and only if for all $\ABCd \in \mathcal{C}$, there exists $\lambda >0$ such that
\begin{align*}
&0 \succ  \left[
\begin{matrix}
\Tr\left\{ \smat{P\\ Y \\ (I_m \otimes \bar{x}) Y + (\bar{u} \otimes I_n)P \\ 0}^\top \ABCd^\top  \right\} \\
\smat{0\\0 \\ I_m \otimes P\\ 0}^\top \ABCd^\top \\
\lambda Y 
\end{matrix}
\right. \\
& \hspace*{4cm} \left.
\begin{matrix}
\dots & \star & \star \\
\dots & - \lambda (I_m \otimes P) & \star \\
\dots & 0 & - \lambda I
\end{matrix}
\right].
\end{align*}
By the parametrization of $\mathcal{C}$ in~\eqref{set C form 3} and algebraic computations, this holds if and only if for all $\Upsilon$ with $\| \Upsilon \| \le 1$, there exists $\lambda > 0$ such that
\begingroup
\setlength\arraycolsep{2pt}
\begin{align*}
0 & \succ 
\bmat{ \Tr\left\{ \smat{P\\ Y \\ (I_m \otimes \bar{x}) Y + (\bar{u} \otimes I_n)P \\ 0}^\top \Zc \right\} & \star & \star \\
\smat{0\\0 \\ I_m \otimes P\\ 0}^\top \Zc & - \lambda (I_m \otimes P) & \star \\
\lambda Y & 0 & - \lambda I} \\
& +
\Tr 
\left\{
\bmat{
\smat{P\\ Y \\ (I_m \otimes \bar{x}) Y + (\bar{u} \otimes I_n)P \\ 0}^\top \\
\smat{0\\0 \\ I_m \otimes P \\ 0}^\top\\
0
}
\mb{A}^{-1/2} \Upsilon \mb{Q}^{1/2}
\bmat{I & 0 & 0}
\right\}.
\end{align*}
\endgroup
By using a common $\lambda >0$, this holds if there exists $\lambda >0$ such that for all $\Upsilon$ with $\| \Upsilon \| \le 1$, the previous matrix inequality is satisfied.
By Petersen's lemma, this holds if and only if there exist $\lambda > 0$ and $\Lambda > 0$ such that
\begingroup
\setlength\arraycolsep{2pt}
\begin{align*}
& 0 \succ 
\bmat{ \Tr\left\{ \smat{P\\ Y \\ (I_m \otimes \bar{x}) Y + (\bar{u} \otimes I_n)P \\ 0}^\top \Zc \right\} & \star & \star \\
\smat{0\\0 \\ I_m \otimes P\\ 0}^\top \Zc & - \lambda (I_m \otimes P) & \star \\
\lambda Y & 0 & - \lambda I} \\
& +
\frac{1}{\Lambda}
\bmat{
\!\smat{P\\ Y \\ (I_m \otimes \bar{x}) Y + (\bar{u} \otimes I_n)P \\ 0}^\top \\
\smat{0\\0 \\ I_m \otimes P \\ 0}^\top\!\\
0
}
\!\!\!\mb{A}^{-1}\!\!\!
\bmat{
\!\smat{P\\ Y \\ (I_m \otimes \bar{x}) Y + (\bar{u} \otimes I_n)P \\ 0}^\top \\
\smat{0\\0 \\ I_m \otimes P \\ 0}^\top\!\\
0
}^\top \\
& +
\Lambda
\smat{I\\ 0\\ 0}
\mb{Q}^{1/2}  \mb{Q}^{1/2}
\smat{I\\ 0\\ 0}^\top.
\end{align*}
\endgroup
By Schur complement and $\Lambda > 0$, this matrix inequality is equivalent to \eqref{sol-bil:qlmi:ct}.
In summary, we have shown that \eqref{bil-decrease-Lyap:ct} holds if \eqref{sol-bil:ct} is feasible and the statement is proven.
\end{proof}

In Theorem~\ref{thm:known_u_bar:ct}, we have the products $\lambda Y$ and $\lambda (I_m \otimes P)$ between decision variables.
This is commonly dealt with by a line search  in the scalar variable $\lambda$ and, for fixed $\lambda$, \eqref{sol-bil:qlmi} corresponds to a linear matrix inequality.

\subsection{Discrete time}

In discrete time, \eqref{sys_tilde} simplifies, by~\eqref{equilibrium_eq}, to
\begin{align}
& \tilde x^+ = 
\ABC[\star] \smat{I\\ K\\ (I_m \otimes \bar{x}) K + \bar{u} \otimes I_n + (I_m \otimes \tilde{x}) K} \tilde{x}. \label{sys_dt_tilde_known_u_bar}
\end{align}
We state the discrete-time counterpart of Theorem~\ref{thm:known_u_bar:ct} next.

\begin{theorem}
\label{thm:known_u_bar:dt}
Under Assumption~\ref{ass:pers ext-bil}, let data $X_1$, $X_0$, $U_0$, $S_0$ in \eqref{data_mat} yield $\mb{A}$, $\Zc$, $\mb{Q}$ in \eqref{set C form 1:ABC}, \eqref{set C form 2:Zc and Q} and let a $\bar{u}$ satisfying~\eqref{equilibrium_eq} be known.
Suppose the next program is feasible
\begin{subequations}
\label{sol-bil}
\begingroup
\setlength\arraycolsep{2.pt}
\begin{align}
& \text{find} & & Y,~P=P^\top \succ 0,~\lambda >0,~\Lambda >0 \label{sol-bil:find}\\
& \text{s.t.}  & & \hspace*{-10pt}
\left[
\begin{matrix}
-P & \star & \\
\Zc^\top \smat{P\\ Y \\ (I_m \otimes \bar{x}) Y + (\bar{u} \otimes I_n)P \\ 0} & -P  & \\
0 & \smat{0\\0 \\ I_m \otimes P\\ 0}^\top \Zc   & \\
\lambda Y & 0   & \\
\mb{A}^{-1/2} \smat{P\\ Y \\ (I_m \otimes \bar{x}) Y + (\bar{u} \otimes I_n)P \\ 0} & 0  & \\
0 & \Lambda \mb{Q}^{1/2}  & \\
\end{matrix}
\right. \notag \\
& & & \left.
\begin{matrix}
\dots & \star & \star & \star & \star \\
\dots & \star & \star & \star & \star \\
\dots & - \lambda (I_m \otimes P) & \star & \star & \star \\
\dots & 0 & - \lambda I & \star & \star \\
\dots & \mb{A}^{-1/2} \smat{0\\0 \\ I_m \otimes P \\ 0} & 0 & - \Lambda I & \star \\
\dots & 0 & 0 & 0 & - \Lambda I
\end{matrix}
\right]
\prec 0 \label{sol-bil:qlmi}
\end{align}
\endgroup
\end{subequations}
and let $P$ and $K = Y P^{-1}$ be a solution to it.
For
\begin{align*}
x^+ \!= A_\star x + B_\star u + C_\star (I_m \otimes x) u + d_\star, u = K(x - \bar{x}) + \bar{u},
\end{align*}
$\bar{x}$ is locally asymptotically stable with basin of attraction including $\mathcal{B}_P^{\bar{x}} :=\{ x \in \real^n \colon (x-\bar{x})^\top P^{-1} (x-\bar{x}) \le 1 \}$.
\end{theorem}
\begin{proof}
\begingroup
\setlength\arraycolsep{1.pt}
Since $\bar{u}$ satisfies \eqref{equilibrium_eq} and by the change of coordinates $\tilde{x} := x - \bar{x}$, the statement becomes that for~\eqref{sys_dt_tilde_known_u_bar}, the origin is locally asymptotically stable with basin of attraction including
$\mathcal{B}_P:=\{ \tilde x \in \real^n \colon \tilde x^\top P^{-1} \tilde x \le 1 \}$.
To prove this, we need to show that, for $P=P^\top \succ 0$, the update of the Lyapunov function $\tilde{x}^\top P^{-1} \tilde{x}$ along solutions to~\eqref{sys_dt_tilde_known_u_bar} satisfies for all $\tilde{x} \in \mathcal{B}_P \backslash \{0\}$
\begin{align*}
& [\star]^\top P^{-1} \cdot \ABC[\star] \smat{I\\ K\\ (I_m \otimes \bar{x}) K + \bar{u} \otimes I_n + (I_m \otimes \tilde{x}) K} \tilde{x} \\
& - \tilde{x}^\top P^{-1} \tilde{x} < 0.
\end{align*}
Since $\ABC[\star]$ is unknown to us, we show instead that for all $\ABCd \in \mathcal{C}$ and all $\tilde{x} \in \mathcal{B}_P \backslash \{ 0\}$,
\begin{align}
& [\star]^\top P^{-1} \cdot \ABC \smat{I\\ K\\ (I_m \otimes \bar{x}) K + \bar{u} \otimes I_n + (I_m \otimes \tilde{x}) K} \tilde{x} \notag \\
& - \tilde{x}^\top P^{-1} \tilde{x} < 0.
\label{bil-decrease-Lyap:dt}
\end{align}
This holds if for all $\ABCd \in \mathcal{C}$, all $\tilde{x} \in \mathcal{B}_P \backslash \{0\}$ and all $\tilde{y} \in \mathcal{B}_P$,
\begin{align*}
& [\star]^\top P^{-1} \cdot \ABC \smat{I\\ K\\ (I_m \otimes \bar{x}) K + \bar{u} \otimes I_n + (I_m \otimes \tilde{y}) K} \tilde{x} \notag \\
& - \tilde{x}^\top P^{-1} \tilde{x} < 0.
\end{align*}
This holds if and only if for all $\ABCd \in \mathcal{C}$ and all $\tilde{y} \in \mathcal{B}_P$,
\begin{align*}
& [\star]^\top P^{-1} \cdot \ABC \smat{I\\ K\\ (I_m \otimes \bar{x}) K + \bar{u} \otimes I_n + (I_m \otimes \tilde{y}) K} - P^{-1} \prec 0.
\end{align*}
By pre/post-multiplying by $P \succ 0$ and setting $K P = Y$, this holds if and only if for all $\ABCd \in \mathcal{C}$ and all $\tilde{y} \in \mathcal{B}_P$,
\begin{align*}
& 0 \succ [\star]^\top P^{-1} \cdot \ABC \smat{P\\ Y\\ (I_m \otimes \bar{x}) Y + (\bar{u} \otimes I_n)P + (I_m \otimes \tilde{y}) Y} - P
\end{align*}
or, equivalently,
\begin{align*}
0 & \succ \bmat{- P & \star\\ 
\ABC \smat{P\\ Y\\ (I_m \otimes \bar{x}) Y + (\bar{u} \otimes I_n)P + (I_m \otimes \tilde{y}) Y}  & - P} \\
& = \bmat{- P & \star\\ 
\ABC \smat{P\\ Y\\ (I_m \otimes \bar{x}) Y + (\bar{u} \otimes I_n)P }  & - P} \\
& +
\Tr \bigg\{
\bmat{0\\ \ABC \smat{0\\ 0\\ I_{mn}} } (I_m \otimes \tilde{y} ) \bmat{Y & 0}
\bigg\}
\end{align*}
by $P \succ 0$, a Schur complement and algebraic computations.
By $P \succ 0$, $\tilde{y} \tilde{y}^\top \preceq P$ is the same as $\tilde{y} \in \mathcal{B}_P$ and $\tilde{y} \tilde{y}^\top \preceq P$ is equivalent to $(I_m \otimes \tilde{y})(I_m \otimes \tilde{y})^\top \preceq I_m \otimes P$.
Then, the previous condition holds if for all $\ABCd \in \mathcal{C}$ and for all $V$ with $V^\top V \preceq I_m  \otimes P$,
\begin{align*}
& 0 \succ
\bmat{- P & \star\\ 
\ABC \smat{P\\ Y\\ (I_m \otimes \bar{x}) Y + (\bar{u} \otimes I_n)P }  & - P} \\
& +
\Tr \bigg\{
\bmat{0\\ \ABC \smat{0\\ 0\\ I_{mn}} } V^\top \bmat{Y & 0}
\bigg\}.
\end{align*}
By Petersen's lemma as reported in \cite[Fact~1]{bisoffi2022petersen} and $P \succ 0$, this holds if and only if for all $\ABCd \in \mathcal{C}$, there exists $\lambda > 0$ such that
\begin{align*}
& 0 \succ \bmat{- P & \star\\ 
\ABC \smat{P\\ Y\\ (I_m \otimes \bar{x}) Y + (\bar{u} \otimes I_n)P }  & - P}  + \lambda \bmat{Y & 0}^\top \bmat{Y & 0} \\
& +  \frac{1}{\lambda} \bmat{0\\ \ABC \smat{0\\ 0\\ I_{mn}} } (I_m \otimes P) \bmat{0\\ \ABC \smat{0\\ 0\\ I_{mn}} }^\top.
\end{align*}
By Schur complement, $\lambda > 0$ and $P \succ 0$, this holds if and only if for all $\ABCd \in \mathcal{C}$, there exists $\lambda >0$ such that
\begin{align*}
0 & \! \succ \!\!\!\!\!
\left[
\begin{matrix}
- P & \star \\
\ABCd \!\! \smat{P\\ Y\\ (I_m \otimes \bar{x}) Y + (\bar{u} \otimes I_n)P \\ 0}  & - P\\
0 & \smat{0\\0 \\ I_m \otimes P\\ 0}^\top \!\! \ABCd^\top \\
\lambda Y & 0 
\end{matrix}
\right. \\
& \left.
\begin{matrix}
\dots & \star & \star\\
\dots & \star & \star\\
\dots & - \lambda (I_m \otimes P) & \star \\
\dots & 0 & - \lambda I
\end{matrix}
\right]
.
\end{align*}
By the parametrization of $\mathcal{C}$ in~\eqref{set C form 3} and algebraic computations, this holds if and only if for all $\Upsilon$ with $\| \Upsilon \| \le 1$, there exists $\lambda > 0$ such that
\begin{align*}
0 & \succ 
\bmat{ 
-P & \star & \star & \star \\
\Zc^\top \smat{P\\ Y\\ (I_m \otimes \bar{x}) Y + (\bar{u} \otimes I_n)P \\ 0}  & - P & \star & \star \\
0 & \smat{0\\0 \\ I_m \otimes P\\ 0}^\top \Zc & - \lambda (I_m \otimes P) & \star \\
\lambda Y & 0 & 0 & - \lambda I} \\
& +
\Tr 
\left\{
\bmat{
\smat{P\\ Y \\ (I_m \otimes \bar{x}) Y + (\bar{u} \otimes I_n)P \\ 0}^\top \\
0\\
\smat{0\\0 \\ I_m \otimes P \\ 0}^\top\\
0
}
\mb{A}^{-1/2} \Upsilon \mb{Q}^{1/2}
\bmat{0 & I & 0 & 0}
\right\}.
\end{align*}
By using a common $\lambda >0$, this holds if there exists $\lambda >0$ such that for all $\Upsilon$ with $\| \Upsilon \| \le 1$, the previous matrix inequality is satisfied.
By Petersen's lemma, this holds if and only if there exist $\lambda > 0$ and $\Lambda > 0$ such that
\begin{align*}
& 0 \succ
\bmat{ 
-P & \star & \star & \star \\
\Zc^\top \smat{P\\ Y\\ (I_m \otimes \bar{x}) Y + (\bar{u} \otimes I_n)P \\ 0}  & - P & \star & \star \\
0 & \smat{0\\0 \\ I_m \otimes P\\ 0}^\top \Zc & - \lambda (I_m \otimes P) & \star \\
\lambda Y & 0 & 0 & - \lambda I} \\
& +
\frac{1}{\Lambda}
\bmat{
\smat{P\\ Y \\ (I_m \otimes \bar{x}) Y + (\bar{u} \otimes I_n)P \\ 0}^\top \\
0\\
\smat{0\\0 \\ I_m \otimes P \\ 0}^\top\\
0
}
\!\!\!\mb{A}^{-1}\!\!\!
\bmat{
\smat{P\\ Y \\ (I_m \otimes \bar{x}) Y + (\bar{u} \otimes I_n)P \\ 0}^\top \\
0\\
\smat{0\\0 \\ I_m \otimes P \\ 0}^\top\\
0
}^\top \\
& +
\Lambda
\smat{0\\ I\\ 0\\ 0}
\mb{Q}^{1/2}  \mb{Q}^{1/2}
\smat{0\\ I\\ 0\\ 0}^\top.
\end{align*}
By Schur complement and $\Lambda > 0$, this matrix inequality is equivalent to \eqref{sol-bil:qlmi}.
In summary, we have shown that \eqref{bil-decrease-Lyap:dt} holds if \eqref{sol-bil} is feasible and the statement is proven.
\endgroup
\end{proof}

For $\bar{x}=0$ and $\bar{u}=0$, our setting reduces to the one considered in \cite{bisoffi2020bilinear}.
We then examine the improvements of Theorem~\ref{thm:known_u_bar:dt}, for $\bar{x}=0$ and $\bar{u}=0$, upon~\cite{bisoffi2020bilinear}.

\begin{remark}[Comparison with \cite{bisoffi2020bilinear}]
\label{remark:comparison}
The common features of Theorem~\ref{thm:known_u_bar:dt}, for $\bar{x}=0$ and $\bar{u}=0$, and \cite{bisoffi2020bilinear} are: (i)~the \emph{scalar} decision variable $\lambda$ appears in products with decision variables $P$ and $Y$, so one fixes $\lambda$ to solve a linear matrix inequality and, on top of that, performs a line search in $\lambda$; (ii)~one can maximize the size of the guaranteed basin of attraction $\mathcal{B}_P^{\bar{x}}$ as in \cite[Cor.~1]{bisoffi2020bilinear} or impose exponential convergence as in \cite[Cor.~2]{bisoffi2020bilinear}.
The improvements of Theorem~\ref{thm:known_u_bar:dt} are: (i)~full development of the multi-input case, see also \cite[Remark~2]{bisoffi2020bilinear}; (ii)~most importantly, the fact that Theorem~\ref{thm:known_u_bar:dt} considers the realistic case of noisy data (for process noise) whereas \cite[Prop.~1]{bisoffi2020bilinear} considers noise-free data and assumes an upperbound on $\| C_\star \|$ and \cite[Prop.~2]{bisoffi2020bilinear} considers noisy data (for measurement noise) and assumes an upperbound on $\| C_\star \|$ and $\| A_\star \|$.
Doing without upperbounds on $\| C_\star \|$ and $\| A_\star \|$ is enabled by applying Petersen's lemma and is very desirable since it reduces conservatism.
\end{remark}

\section{Data-driven setpoint control with unknown equilibrium input}
\label{sec:UNknown_u_bar}

In this section we address Problem~\ref{problem:statement} when we do \emph{not} know a value $\bar{u}$ such that $(\bar{x},\bar{u})$ is an equilibrium for~\eqref{sys_alt}.
Then, we treat $\bar{u}$ as a design parameter, along with $K$, of the control law in~\eqref{control_law}.
In continuous time, \eqref{sys_tilde} becomes
\begin{align}
& \dot{\tilde x} = 
\ABC[\star] \smat{I\\ K\\ (I_m \otimes \bar{x})K + \bar{u} \otimes I_n + (I_m \otimes \tilde{x}) K} \tilde{x} \notag \\
& +
\ABCd[\star]\smat{\bar{x}\\ \bar{u}\\ (I_m \otimes \bar{x}) \bar{u}\\ 1}. \label{sys_ct_tilde_2}
\end{align}
In Section~\ref{sec:known_u_bar}, the known equilibrium input $\bar{u}$ ensured  $\ABCd[\star] \smat{\bar{x}\\ \bar{u}\\ (I_m \otimes \bar{x}) \bar{u}\\ 1} = 0$.
Since $\ABCd[\star]$ is unknown, we would be tempted to design $\bar{u}$ so that
\begin{align}
\label{ideal_Gamma_for_all_ABC_CT}
& \ABCd \smat{\bar{x}\\ \bar{u}\\ (I_m \otimes \bar{x}) \bar{u}\\ 1}  = 0 \,\,\,\, \forall \ABCd \in \mathcal{C}.
\end{align}
We note that for $\mb{Q} = 0$, the set $\mathcal{C}$ in~\eqref{set C form 2} reduces, by Assumption~\ref{ass:pers ext-bil}, to the singleton $\mathcal{C} = \{ \Zc^\top \} = \ABCd[\star]$ and the actual system could be identified, so we focus on the interesting case of $\mb{Q} \neq 0$.
Under this setting, however, the next result shows that \eqref{ideal_Gamma_for_all_ABC_CT} is impossible to achieve.

\begin{lemma} 
\label{lemma:cannot_find_u_bar}
Under Assumption~\ref{ass:pers ext-bil} and $\mb{Q} \neq 0$, there does not exist $\bar{u}$ such that \eqref{ideal_Gamma_for_all_ABC_CT} holds.
\end{lemma}
\begin{proof}
Under Assumption~\ref{ass:pers ext-bil}, \eqref{ideal_Gamma_for_all_ABC_CT} holds if and only if
\begin{align*}
0  = \smat{\bar{x}\\ \bar{u}\\ (I_m \otimes \bar{x}) \bar{u}\\ 1}^\top (\Zc + \mb{A}^{-1/2} \Upsilon \mb{Q}^{1/2} ) \quad \forall \Upsilon \colon \| \Upsilon \| \le 1. 
\end{align*}
By considering $\Upsilon = 0$, this holds if and only if
\begin{equation}
\label{equiv_cond_Gamma_CT}
\begin{aligned}
0 & = \smat{\bar{x}\\ \bar{u}\\ (I_m \otimes \bar{x}) \bar{u}\\ 1}^\top \Zc , \\
0 & = \smat{\bar{x}\\ \bar{u}\\ (I_m \otimes \bar{x}) \bar{u}\\ 1}^\top \mb{A}^{-1/2} \Upsilon \mb{Q}^{1/2}  \quad \forall \Upsilon \colon \| \Upsilon \| \le 1. 
\end{aligned}
\end{equation}
Since $\mb{Q}^{1/2}=(\mb{Q}^{1/2})^\top \succeq 0$, its $n$ eigenvalues can be ordered as $\lambda_1 \ge \dots \ge \lambda_n \ge 0$ and, for $\Lambda := \diag(\lambda_1,\dots, \lambda_n)$, there exists a real orthogonal matrix $T$ (i.e., $T^\top T = T T^\top = I$) such that
\begin{equation*}
\mb{Q}^{1/2} = T \Lambda T^\top,
\end{equation*}
which is an eigendecomposition of $\mb{Q}^{1/2}$.
Moreover, since $\mb{Q} \neq 0$, we have at least $\lambda_1 > 0$.
Hence, the second condition in \eqref{equiv_cond_Gamma_CT} rewrites
\begin{align}
\label{equiv_cond_Gamma_2_CT}
0 & = \smat{\bar{x}\\ \bar{u}\\ (I_m \otimes \bar{x}) \bar{u}\\ 1}^\top \mb{A}^{-1/2} \Upsilon T \Lambda T^\top  \quad \forall \Upsilon \colon \| \Upsilon \| \le 1.
\end{align}
By $\mb{A}^{-1/2} \succ 0$, we have $\mb{A}^{-1/2} \smat{\bar{x}\\ \bar{u}\\ (I_m \otimes \bar{x}) \bar{u}\\ 1}  =: \bar z \neq 0$ for all $\bar{u}$. Suppose then that the $i$-th component $\bar{z}_i$ of $\bar{z}$ is nonzero.
Select $\Upsilon^\top = T \smat{0 & \dots & e_1 & \dots & 0}$ where $e_1 :=(1,0,\dots,0)$ is the $i$-th and only nonzero column of $\smat{0 & \dots & e_1 & \dots & 0}$.
We have $\Upsilon \Upsilon^\top = \diag (0, \dots, 0, 1,0, \dots, 0)$ with the $1$ in position $(i,i)$ and, thus, $\Upsilon \Upsilon^\top \preceq I$.
For such $\Upsilon$, \eqref{equiv_cond_Gamma_2_CT} implies
\begin{align*}
0 = \bar{z}^\top \smat{0 & \dots & e_1 & \dots & 0}^\top \Lambda T^\top = \smat{\lambda_1 \bar{z}_i & 0 & \dots  & 0} T^\top,
\end{align*}
which cannot hold since $\bar{z}_i \neq 0$ and $\lambda_1 > 0$.
In other words, we have shown that for all $\bar{u}$, there exists $\Upsilon$ with $\| \Upsilon \| \le 1$ such that $0 \neq \smat{\lambda_1 \bar{z}_i & 0 & \dots  & 0}$, and this proves the statement.
\end{proof}

Motivated by the negative result in Lemma~\ref{lemma:cannot_find_u_bar}, we then relax \eqref{ideal_Gamma_for_all_ABC_CT} to
\begin{subequations}
\label{min_progr_Gamma_CT}
\begin{align}
& \text{min.} & & \gamma \quad \text{(over $\gamma$, $\bar{u}$)} \label{min_progr_Gamma_CT:min}\\
& \text{s. t.} & & \bigg| \ABCd \smat{\bar{x}\\ \bar{u}\\ (I_m \otimes \bar{x}) \bar{u}\\ 1} \bigg|^2 \le \gamma \,\,\,\, \forall \ABCd \in \mathcal{C}. \label{min_progr_Gamma_CT:for_all}
\end{align}
\end{subequations}
We show next that \eqref{min_progr_Gamma_CT} is actually equivalent to a feasible semidefinite program.
\begin{lemma}
\label{lemma:design_u_bar_gamma}
Under Assumption~\ref{ass:pers ext-bil}, let data $X_1$, $X_0$, $U_0$, $S_0$ in \eqref{data_mat} yield $\mb{A}$, $\Zc$, $\mb{Q}$ in \eqref{set C form 1:ABC}, \eqref{set C form 2:Zc and Q} and suppose $\mb{Q} \neq 0$.
Then, \eqref{min_progr_Gamma_CT} is equivalent to
\begingroup
\setlength\arraycolsep{2pt}
\begin{subequations}
\label{min_progr_Gamma_bil_equiv_CT}
\begin{align}
& \text{min.}\hspace*{-3mm} & & \gamma \quad \text{(over $\gamma$, $\bar{u}$, $\sigma$)} \label{min_progr_Gamma_bil_equiv_CT:min}\\
& \text{s.t.}\hspace*{-5mm} & & 
\bmat{
- \gamma I & \star & \star & \star\\
\smat{\bar{x}\\ \bar{u}\\ (I_m \otimes \bar{x}) \bar{u}\\ 1}^\top \Zc & -I & \star & \star\\
0 & \mb{A}^{-1/2} \smat{\bar{x}\\ \bar{u}\\ (I_m \otimes \bar{x}) \bar{u}\\ 1} & -\sigma I & \star\\
\sigma \mb{Q}^{1/2} & 0 & 0 & - \sigma I} \preceq 0 ,\notag\\
& & & \sigma >0 \label{min_progr_Gamma_bil_equiv_CT:for_all}
\end{align}
\end{subequations}
\endgroup
and is feasible.
\end{lemma}
\begin{proof}
For brevity, we use $\ABCd = Z^\top \in \mathcal{C}$ as in~\eqref{set C form 1:set only} and rewrite \eqref{min_progr_Gamma_CT} as 
\begin{subequations}
\label{min_progr_Gamma_compact_CT}
\begin{align}
& \text{min.} & & \gamma \quad \text{(over $\gamma$, $\bar{u}$)} \label{min_progr_Gamma_compact_CT:min}\\
& \text{s. t.} & & | Z^\top \bar{\nu}(\bar{u},\bar x) |^2 := 
\bigg| Z^\top \smat{\bar{x}\\ \bar{u}\\ (I_m \otimes \bar{x}) \bar{u}\\ 1} \bigg|^2 \le \gamma \,\,\,\, \forall Z^\top \in \mathcal{C}.  \label{min_progr_Gamma_compact_CT:for_all}
\end{align}
\end{subequations}
\eqref{min_progr_Gamma_compact_CT:for_all} is equivalent to
\begin{align*}
\big(Z^\top \bar{\nu}(\bar{u},\bar x) \big) \big( Z^\top \bar{\nu}(\bar{u},\bar x)  \big)^\top
\preceq \gamma I \quad \forall Z^\top \in \mathcal{C}
\end{align*}
and, by Schur complement, to
\begin{align*}
\bmat{- \gamma I & Z^\top \nu(\bar{u},\bar x) \\ 
\nu(\bar{u},\bar x)^\top  Z & -I}
\preceq 0 \quad \forall Z^\top \in \mathcal{C}
\end{align*}
and, by Assumption~\ref{ass:pers ext-bil}, to
\begin{align*}
& \forall \Upsilon \colon \| \Upsilon \| \le 1,~~ 0 \succeq 
\bmat{- \gamma I & ~(\Zc + \mb{A}^{-1/2} \Upsilon \mb{Q}^{1/2} )^\top \nu(\bar{u},\bar x) \\ 
\star  & -I} \\
& = 
\bmat{- \gamma I & ~\Zc^\top \nu(\bar{u},\bar x) \\ \star & -I} \\
& \quad   + 
\Tr \left\{ 
\bmat{0\\ \nu(\bar{u},\bar x)^\top} \mb{A}^{-1/2} \Upsilon \mb{Q}^{1/2} \bmat{I & 0}
\right\}
\end{align*}
and, by nonstrict Petersen's lemma in the version reported in \cite[Fact~2]{bisoffi2022petersen} and $\mb{Q} \neq 0$ and Assumption~\ref{ass:pers ext-bil}
~\footnote{
To apply \cite[Fact~2]{bisoffi2022petersen}, we need to verify $\bmat{0\\ \nu(\bar{u},\bar x)^\top} \mb{A}^{-1/2} \neq 0$ and $ \mb{Q}^{1/2} \bmat{I & 0} \neq 0$ or, equivalently, $\mb{A}^{-1/2} \smat{0\\ \bar{x}\\ \bar{u}\\ (I_m \otimes \bar{x}) \bar{u}\\ 1} \neq 0$ and $\mb{Q}^{1/2}\neq 0$. These conditions are true because $\mb{A} \succ 0$ by Assumption~\ref{ass:pers ext-bil} and $\mb{Q} \neq 0$.
}, to
\begin{align*}
& \exists \sigma > 0 \colon  0 \succeq  
\bmat{- \gamma I &  \Zc^\top \nu(\bar{u},\bar x)\\ 
\star  & -I} 
+ \sigma \bmat{I\\ 0} \mb{Q}^{1/2} \mb{Q}^{1/2} \bmat{I\\ 0}^\top \\
& + \frac{1}{\sigma}\bmat{0\\ \nu(\bar{u},\bar x)^\top} \mb{A}^{-1/2} \mb{A}^{-1/2} 
\bmat{0\\ \nu(\bar{u},\bar x)^\top}^\top 
\end{align*}
and, by Schur complement and $\sigma > 0$, to \eqref{min_progr_Gamma_bil_equiv_CT:for_all}.
This shows the first part of the statement, i.e., the equivalence of \eqref{min_progr_Gamma_compact_CT} and \eqref{min_progr_Gamma_bil_equiv_CT}.
We then show that \eqref{min_progr_Gamma_bil_equiv_CT} is feasibile because \eqref{min_progr_Gamma_compact_CT} is feasible.
Under Assumption \ref{ass:pers ext-bil}, the same steps in \cite[Lemma~2]{bisoffi2022petersen} show also for the bilinear case under consideration that, for each $Z^\top \in \mathcal{C}$, $\| Z \| \le \bar{z}_{\|\|} := \| \Zc \| + \lambda_{\min}(\mb{A})^{-1/2} \| \mb{Q}^{1/2} \|$.
A feasible $(\gamma,\bar{u})$ for \eqref{min_progr_Gamma_compact_CT} is given by $\gamma = \bar{z}_{\|\|}^2 (|\bar{x}|^2+1)$ and $\bar{u} = 0$. Indeed, this $(\gamma,\bar{u})$ ensures, for all $Z^\top \in \mathcal{C}$, $| Z^\top \bar{\nu}(0,\bar x) |^2 \le \bar{z}_{\|\|}^2 (|\bar{x}|^2 + 1) = \gamma$, i.e., satisfaction of~\eqref{min_progr_Gamma_compact_CT:for_all}.
\end{proof}

\begin{figure*}
\begin{subequations}
\label{progr_K_bil_setpoint_CT}
\begin{align}
& \text{find}  & & P \succ 0, Y, s, \tau_\eta \ge 0, \tau_\gamma \ge 0, \lambda > 0, \Lambda > 0 \label{progr_K_bil_setpoint_CT:find}\\
& \text{s.t.} & & s+\tau_\eta \ge 0,  \tau_\eta (\eta-1) = \epsilon + s + \tau_\gamma \gamma \label{progr_K_bil_setpoint_CT:constr1}\\
&  & & 
\bmat{
\Tr\bigg\{ \smat{P\\ Y\\ (I_m \otimes \bar{x})Y + (\bar{u} \otimes I_n)P\\ 0 }^\top \Zc \bigg\} - s P & \star & \star & \star & \star & \star \\
I_n & - \tau_\gamma I_n & \star & \star & \star & \star \\
\smat{0\\ 0\\  I_m \otimes P \\ 0}^\top \Zc & 0 & -\lambda (I_m \otimes P) & \star & \star & \star\\
\lambda Y & 0 & 0 & - \lambda I_m & \star & \star\\
\mb{A}^{-1/2} \smat{P\\ Y\\ (I_m \otimes \bar{x})Y + (\bar{u} \otimes I_n)P\\ 0 } & 0 & \mb{A}^{-1/2} \smat{0\\ 0\\  I_m \otimes P \\ 0} & 0 & - \Lambda I_{n+m+mn+1} & \star\\
\Lambda \mb{Q}^{1/2} & 0 & 0 & 0 & 0 & - \Lambda I_n
} \preceq 0 \label{progr_K_bil_setpoint_CT:constr2} 
\end{align}
\end{subequations}
\hrulefill
\end{figure*}

Since \eqref{min_progr_Gamma_bil_equiv_CT} is feasible by Lemma~\ref{lemma:design_u_bar_gamma}, it always returns a pair $(\gamma,\bar{u})$ under the given hypothesis, and we use this $(\gamma,\bar{u})$ to design the control gain $K$ as in the next result.

\begin{theorem}
\label{thm:UNknown_u_bar}
Under the same hypothesis of Lemma~\ref{lemma:design_u_bar_gamma}, obtain $(\gamma,\bar{u})$ from solving \eqref{min_progr_Gamma_bil_equiv_CT}.
For this $(\gamma,\bar{u})$ and given $\epsilon > 0$ and $\eta \in (0,1)$, suppose the program in~\eqref{progr_K_bil_setpoint_CT}, displayed over two columns, is feasible and let $P$ and $K= Y P^{-1}$ be a solution to it.
For
\begin{align*}
\dot x = A_\star x + B_\star u + C_\star (I_m \otimes x) u + d_\star, u = K(x - \bar{x}) + \bar{u},
\end{align*}
the set $\{ x \in \real^n \colon (x- \bar{x})^\top P^{-1} (x- \bar{x}) \le \eta \}$ is locally asymptotically stable with basin of attraction including $\{ x \in \real^n \colon (x- \bar{x})^\top P^{-1} (x- \bar{x}) \le 1 \}$.
\end{theorem}
\begin{proof}
\begingroup
\setlength\arraycolsep{1.pt}
Before proving Theorem~\ref{thm:UNknown_u_bar}, we provide the next auxiliary result, for which, given $P \succ 0$, we define function $V_P \colon \real^n \to \real$ and set $\B[P]$ as
\begin{align}
\label{Lyap_fun_set_of_Q_CT}
V_P (\tilde{x}) := \tilde{x}^\top P^{-1} \tilde{x} \text{ and } \B[P]  & := \{ \tilde x \in \real^n \colon \tilde{x}^\top P^{-1} \tilde{x} \le 1 \}
\end{align}
and $\overline{\mathcal{S}}$ as the closure of a generic set $\mathcal{S}$.
\begin{lemma}
\label{lemma:set_attractive_inv_CT}
Consider the dynamics $\dot{\tilde{x}} = F(\tilde{x})$ with $F$ continuously differentiable.
If, for some $P \succ 0$, $\epsilon >0$ and $\eta \in (0,1)$,
\begin{align}
& \frac{\partial V_P}{\partial \tilde{x}}(\tilde{x}) F(\tilde{x})
\le - \epsilon \quad \forall \tilde{x} \in \overline{\B[P] \backslash \B[\eta P]} \label{decrease_on_annulus}
\end{align}
holds, then (i)~solutions to $\dot{\tilde{x}} = F(\tilde{x})$ with initial condition in $\B[P]$ eventually reach $\B[\eta P]$ and remain in $\B[\eta P]$ thereafter; (ii)~the set $\B[\eta P]$ is locally asymptotically stable with basin of attraction including $\B[P]$. 
\end{lemma}
\begin{proofof}{Lemma~\ref{lemma:set_attractive_inv_CT}}
Note that, by $\eta \in (0,1)$, $\B[\eta P] = \{ \tilde{x} \in \real^n \colon \tilde{x}^\top (\eta^{-1} P^{-1}) \tilde{x} \le 1\} \subseteq \B[P]$. 
We start proving (i).
Suppose by contradiction that there exists a solution $\tilde{x}$ with initial condition in $\B[P]$ that does not reach the set $\B[\eta P]$, i.e., $\tilde{x}(0) \in \mathcal{B}_P$ and $\tilde{x}(t) \notin \B[\eta P]$ for all $t \ge 0$.
It can happen \emph{either} that such solution exits $\B[P]$ at some time $\bar{t} \ge 0$, i.e., $\tilde{x}(t) \in \B[P]$ for all $t$ with $0 \le t \le \bar{t}$ and for some $\bar{\epsilon} > 0$, $\tilde{x}(t) \notin \B[P]$ for all $t$ with $\bar t < t \le \bar t + \bar \epsilon$; \emph{or} that such solution never exits $\B[P]$, i.e., $\tilde{x}(t) \in \B[P]$ for all $t \ge 0$.
In the first case, consider the function $t \mapsto v(t) := V_{P}(\tilde{x}(t))$, which is continuously differentiable since it is a composition of continuously differentiable functions ($\tilde{x}$ is continuously differentiable since it is a solution in the classical sense, as $F$ is continuously differentiable.)
We have that $v(\bar t) = V_{P}(\tilde{x}(\bar t)) \le 1$ and $v(\bar t + h) = V_{P}(\tilde{x}(\bar t + h)) > 1$ for all $h \in (0,\bar{\epsilon}]$; hence, $v(\bar t + h) - v(\bar t) > 0$ for all $h \in (0,\bar{\epsilon}]$ and $\dot{v}(\bar t) := \lim_{h \to 0} \frac{v(\bar t + h) - v(\bar t)}{h} \ge 0$; on the other hand, 
\begin{align*}
\dot{v}(\bar t) = \tfrac{\partial V_{P}}{\partial \tilde x}(\tilde{x}(\bar t) ) \dot{\tilde{x}}(\bar t) = \tfrac{\partial V_{P}}{\partial \tilde x}(\tilde{x}(\bar t) ) F(\tilde{x}(\bar t)) \le - \epsilon
\end{align*}
since $\tilde{x}(\bar{t}) \in \B[P] \backslash \B[\eta P]$, thereby contradicting $\dot{v}(\bar t) \ge 0$.
So, the solution considered in the first case cannot exist.
In the second case, $\tilde{x}(t) \in \B[P]\backslash \B[\eta P]$ for all $t\ge 0$ and, then, such solution satisfies
\begin{align*}
& \tilde{x}(t)^\top P^{-1} \tilde{x}(t) = v(t) = v(0) + \int_0^t \tfrac{\partial V_{P}}{\partial \tilde x}(\tilde{x}(\tau) ) F(\tilde{x}(\tau)) d\tau \\
& \le v(0) - \epsilon t \le 1 - \epsilon t, \quad \forall t \ge 0.
\end{align*}
Taking $t$ sufficiently large contradicts that $\tilde{x}(t)^\top P^{-1} \tilde{x}(t) > \eta$ for all $t \ge 0$.
In summary, since both the first and second case have led to a contradiction, each solution with initial condition in $\B[P]$ must reach the set $\B[\eta P]$ eventually.
The set $\B[\eta P]$ is (forward) invariant because \eqref{decrease_on_annulus} implies that $\frac{\partial V_P}{\partial \tilde{x}}(\tilde{x}) F(\tilde{x}) \le - \epsilon$ for all $\tilde{x}$ with $\tilde{x}^\top P^{-1} \tilde{x} = \eta$.
We end proving (ii).
By $P\succ 0$ and $\eta \in (0,1)$, $\B[\eta P]$ is a compact set and is locally attractive and (forward) invariant as just shown; by regularity of $F$ (continuous differentiability), results such as \cite[Prop.~7.5]{goebel2012hybrid} yield that $\B[\eta P]$ is locally asymptotically stable.
Its basin of attraction includes $\B[P]$ by point~(i).
\end{proofof}

With the change of coordinates $\tilde{x} := x - \bar{x}$ and the definition in~\eqref{Lyap_fun_set_of_Q_CT}, the statement of Theorem~\ref{thm:UNknown_u_bar} is that $\B[\eta P]$ is locally asymptotically stable with basin of attraction including $\B[P]$ for
\begin{align}
&  \dot{\tilde{x}}= \ABCd[\star] \smat{I\\ K\\ (I_m \otimes \bar{x})K + \bar{u} \otimes I_n + (I_m \otimes \tilde{x}) K\\ 0} \tilde{x} \notag \\
& + \ABCd[\star]\smat{\bar{x}\\ \bar{u}\\ (I_m \otimes \bar{x}) \bar{u}\\ 1} \notag \\
& =: \ABCd[\star] \big( \mu(K,\tilde{x}) \tilde{x} + \bar{\nu}(\bar{u},\bar{x}) \big) \label{def_mu_nu}
\end{align}
where $\bar{u}$ is returned by~\eqref{min_progr_Gamma_bil_equiv_CT}.
Since $\ABCd[\star]$ is unknown to us, we show instead, after setting $\ABCd = Z^\top$ as in~\eqref{set C form 1:set only} for brevity, that $\B[\eta P]$ is locally asymptotically stable with basin of attraction including $\B[P]$ for
\begin{align}
&  \dot{\tilde{x}}= Z^\top \big( \mu(K,\tilde{x}) \tilde{x} + \bar{\nu}(\bar{u},\bar{x}) \big)
\end{align}
for all $Z^\top \in \mathcal{C}$.
By Lemma~\ref{lemma:set_attractive_inv_CT}, this statement holds if there exist $P = P^\top \succ 0$ and $K$ such that
\begin{subequations}
\begin{align}
& \forall Z^\top \in \mathcal{C}, \tilde{x} \in \overline{\B[P] \backslash \B[\eta P]}\label{attractivity_CT}\\
& \quad \tilde{x}^\top P^{-1} \dot{\tilde{x}} + \dot{\tilde{x}}^\top P^{-1} \tilde{x} \le -\epsilon, \hspace*{5pt} \dot{\tilde{x}} = Z^\top \big( \mu(K,\tilde{x}) \tilde{x} + \bar{\nu}(\bar{u},\bar{x}) \big) .\notag 
\end{align}
\end{subequations}
For $\gamma \ge 0$ returned by~\eqref{min_progr_Gamma_bil_equiv_CT}, define the set
\begin{align*}
\mathcal{B}_\gamma  & := \{ \bar{\xi} \in \real^n \colon \bar{\xi}^\top \bar{\xi} \le \gamma \}.
\end{align*}
\eqref{attractivity_CT} holds if
\begin{align}
& \forall Z^\top \in \mathcal{C}, \tilde{y} \in \B[P], \tilde{x} \in \overline{\B[P] \backslash \B[\eta P]}, \bar{\xi} \in \mathcal{B}_\gamma \label{attractivity-replace-with-bar-xi_CT} \\
& \quad 
\tilde{x}^\top P^{-1} \dot{\tilde{x}} + \dot{\tilde{x}}^\top P^{-1} \tilde{x} \le -\epsilon, \hspace*{5pt} \dot{\tilde{x}} = Z^\top \mu(K,\tilde{y}) \tilde{x} + \bar{\xi}. \notag
\end{align}
because \eqref{min_progr_Gamma_bil_equiv_CT}, to which $(\gamma,\bar{u})$ is a solution, ensures by construction that $\forall Z^\top \in \mathcal{C}$, $|Z^\top \bar{\nu}(\bar{u},\bar{x}) |^2 \le \gamma$.
By the S-procedure \cite[\S 2.6.3]{boyd1994linear}, the previous condition holds if
\begin{align*}
& \forall Z^\top \in \mathcal{C}, \tilde{y} \in \B[P], \exists \tau_P \ge 0, \tau_\eta \ge 0, \tau_\gamma \ge 0 \colon \forall \tilde{x}, \bar{\xi}\\
& \Tr\{ \tilde{x}^\top P^{-1} ( Z^\top \mu(K,\tilde{y}) \tilde{x} + \bar \xi )  \}  + \epsilon\\
& - \tau_P  \big( \tilde{x}^\top P^{-1} \tilde x - 1 \big)
- \tau_\eta \big( \eta - \tilde x^\top P^{-1} \tilde x \big)
- \tau_\gamma \big( \bar{\xi}^\top \bar{\xi} - \gamma \big)
\le 0.
\end{align*}
This holds, by algebraic computations, if and only if
\begin{align}
& \forall Z^\top \in \mathcal{C}, \tilde{y} \in \B[P], \exists \tau_P \ge 0, \tau_\eta \ge 0, \tau_\gamma \ge 0 \colon \notag \\
& 
\bmat{
\left\{
\begin{matrix}
\Tr\{ P^{-1} Z^\top \mu(K,\tilde{y}) \} \\
- (\tau_P - \tau_\eta ) P^{-1} 
\end{matrix}
\right\}
& \star & \star\\
P^{-1} & - \tau_\gamma I_n & \star\\
0 & 0 & 
\left\{
\begin{matrix}
\epsilon + \tau_P \\
- \tau_\eta \eta + \tau_\gamma \gamma
\end{matrix}
\right\} \\
} \preceq 0 . \label{bil_setpoint_proof_1_CT}
\end{align}
By an argument analogous to that in~\cite[p.~83]{boyd1994linear}, this condition is equivalent to
\begin{align}
& \forall Z^\top \in \mathcal{C}, \tilde{y} \in \B[P], \exists \tau_P \ge 0, \tau_\eta \ge 0, \tau_\gamma \ge 0 \colon \notag \\
& \epsilon + \tau_P - \tau_\eta \eta + \tau_\gamma \gamma = 0 \notag \\
& 
\bmat{
\left\{
\begin{matrix}
\Tr\{ P^{-1} Z^\top \mu(K, \tilde{y}) \}\\
- (\tau_P - \tau_\eta ) P^{-1} 
\end{matrix}
\right\}
& \star \\
P^{-1} & - \tau_\gamma I_n
}\preceq 0. \label{bil_setpoint_proof_2_CT}
\end{align}
To establish this equivalence, we show that \eqref{bil_setpoint_proof_1_CT} implies \eqref{bil_setpoint_proof_2_CT} since the converse is immediate.
Suppose that \eqref{bil_setpoint_proof_1_CT} holds for $\tau_\eta$, $\tau_P$, $\tau_\gamma$ and let us find $\bar \tau_\eta \ge 0$ for the very same $\tau_P$, $\tau_\gamma$ such that \eqref{bil_setpoint_proof_2_CT} holds.
Since \eqref{bil_setpoint_proof_1_CT} holds, the block-diagonal structure yields that $\epsilon + \tau_P - \tau_\eta \eta + \tau_\gamma \gamma \le 0$ and
\begin{align*}
N :=
\bmat{
\left\{
\begin{matrix}
\Tr\{ P^{-1} Z^\top \mu(K, \tilde{y}) \}\\
- (\tau_P - \tau_\eta ) P^{-1} 
\end{matrix}
\right\}
& \star \\
P^{-1} & - \tau_\gamma I_n
}\preceq 0.
\end{align*}
Since $\epsilon + \tau_P + \tau_\gamma \gamma \le \tau_\eta \eta$, we can select $\bar \tau_\eta \ge 0$ such that $\epsilon + \tau_P + \tau_\gamma \gamma = \bar \tau_\eta \eta$, which yields $\bar \tau_\eta \eta \le \tau_\eta \eta$ and $\tau_\eta P^{-1} \succeq \bar \tau_\eta P^{-1}$ by $P \succ 0$.
Hence, 
\begin{align*}
\bmat{
\left\{
\begin{matrix}
\Tr\{ P^{-1} Z^\top \mu(K, \tilde{y}) \}\\
- (\tau_P - \bar \tau_\eta ) P^{-1} 
\end{matrix}
\right\}
& \star \\
P^{-1} & - \tau_\gamma I_n
}\preceq N \preceq 0
\end{align*}
and this shows that \eqref{bil_setpoint_proof_2_CT} holds.
After proving the equivalence of \eqref{bil_setpoint_proof_1_CT} and \eqref{bil_setpoint_proof_2_CT}, we resume the main thread of the proof.
In a standard fashion, pre/post-multiply the matrix inequality in~\eqref{bil_setpoint_proof_2_CT} by a block diagonal matrix with entries $P$, $I$; set $K P = Y$, after recalling the definition of $\mu(K,\tilde{y})$ in~\eqref{def_mu_nu}; the previous condition holds if and only if there exist $P = P^\top \succ 0$ and $Y$ such that
\begin{align*}
& \forall Z^\top \in \mathcal{C}, \tilde{y} \in \B[P], \exists \tau_P \ge 0, \tau_\eta \ge 0, \tau_\gamma \ge 0 \colon \notag \\
& \tau_\eta \eta = \epsilon + \tau_P + \tau_\gamma \gamma \\
& 
0 \succeq \bmat{
\left\{
\begin{matrix}
\Tr \bigg\{ Z^\top \smat{P\\ Y\\ (I_m \otimes \bar{x})Y + (\bar{u} \otimes I_n)P + (I_m \otimes \tilde{y}) Y\\ 0} \bigg\}\\
- (\tau_P - \tau_\eta ) P
\end{matrix}
\right\}
& \star \\
I_n & - \tau_\gamma I_n
}  \\
& \phantom{ 0 } = 
\bmat{
\left\{
\begin{matrix}
\Tr \bigg\{ Z^\top \smat{P\\ Y\\ (I_m \otimes \bar{x})Y + (\bar{u} \otimes I_n)P\\ 0 } \bigg\}\\
- (\tau_P - \tau_\eta ) P
\end{matrix}
\right\}
& \star \\
I_n & - \tau_\gamma I_n
} \\
& \qquad +
\Tr
\left\{
\bmat{
Z^\top  \smat{0\\ 0\\ I_{mn} \\ 0 }\\
0 }
(I_m \otimes \tilde{y})
\bmat{Y & 0}
\right\}.
\end{align*}
By $P \succ 0$, $\tilde y \tilde y^\top \preceq P$ is the same as $\tilde y \in \B[P]$ and $\tilde y \tilde y^\top \preceq P$ is equivalent to $(I_m \otimes \tilde y)(I_m \otimes \tilde y)^\top \preceq I_m \otimes P$.
Then, the previous condition holds if
\begin{align*}
& \exists \tau_P \ge 0, \tau_\eta \ge 0, \tau_\gamma \ge 0, \lambda > 0 \colon \forall Z^\top \in \mathcal{C},  \notag \\
& \tau_\eta \eta = \epsilon + \tau_P + \tau_\gamma \gamma \\
& \left[
\begin{matrix}
\left\{
\begin{matrix}
\Tr\bigg\{ \smat{P\\ Y\\ (I_m \otimes \bar{x})Y + (\bar{u} \otimes I_n)P\\ 0 }^\top Z\bigg\}\\
- (\tau_P - \tau_\eta ) P
\end{matrix}
\right\} & \star \\
I_n & - \tau_\gamma I_n \\
\smat{0\\ 0\\  I_m \otimes P \\ 0}^\top Z  & 0  \\
\lambda Y & 0 
\end{matrix}
\right. \\
& \hspace*{3.5cm}\left.
\begin{matrix}
\dots & \star & \star \\
\dots & \star & \star \\
\dots & -\lambda (I_m \otimes P) & \star \\
\dots & 0 & - \lambda I_m \\
\end{matrix}
\right] \preceq 0
\end{align*}
where we first applied Petersen's lemma as reported in \cite[Fact~2]{bisoffi2022petersen} and then used Schur complement.
By Assumption~\ref{ass:pers ext-bil} and the parametrization of $\mathcal{C}$ in~\eqref{set C form 3}, this holds if and only if
\begin{align*}
& \exists \tau_P \ge 0, \tau_\eta \ge 0, \tau_\gamma \ge 0, \lambda > 0 \colon \forall \Upsilon \text{ with } \| \Upsilon \| \le 1,  \notag \\
& \tau_\eta \eta = \epsilon + \tau_P + \tau_\gamma \gamma \\
& 
0 \succeq \left[
\begin{matrix}
\left\{
\begin{matrix}
\Tr\bigg\{ \smat{P\\ Y\\ (I_m \otimes \bar{x})Y + (\bar{u} \otimes I_n)P\\ 0 }^\top\Zc  \bigg\}\\
- (\tau_P - \tau_\eta ) P
\end{matrix}
\right\} & \star \\
I_n & - \tau_\gamma I_n \\
\smat{0\\ 0\\  I_m \otimes P \\ 0}^\top \Zc  & 0  \\
\lambda Y & 0 
\end{matrix}
\right.\\
& \hspace*{1.cm}\left.
\begin{matrix}
\dots & \star & \star \\
\dots & \star & \star \\
\dots & -\lambda (I_m \otimes P) & \star\\
\dots & 0 & - \lambda I_m
\end{matrix}
\right]
\\
&
+ \Tr
\left\{
\bmat{\smat{P\\ Y\\ (I_m \otimes \bar{x})Y + (\bar{u} \otimes I_n)P\\ 0 }^\top\\ 
0\\
\smat{0\\ 0\\  I_m \otimes P \\ 0}^\top\\
0}
\mb{A}^{-1/2} \Upsilon \mb{Q}^{1/2}
\bmat{I_n & 0 & 0 & 0}
\right\}.
\end{align*}
By Petersen's lemma, this holds if
\begin{align*}
& \exists \tau_P \ge 0, \tau_\eta \ge 0, \tau_\gamma \ge 0, \lambda > 0, \Lambda > 0 \colon \\
& \tau_\eta \eta = \epsilon + \tau_P + \tau_\gamma \gamma, \\
& 
\left[
\begin{matrix}
\left\{
\begin{matrix}
\Tr\bigg\{ \smat{P\\ Y\\ (I_m \otimes \bar{x})Y + (\bar{u} \otimes I_n)P\\ 0 }^\top \Zc \bigg\}\\
- (\tau_P - \tau_\eta ) P
\end{matrix}
\right\} & \star \\
I_n & - \tau_\gamma I_n  \\
\smat{0\\ 0\\  I_m \otimes P \\ 0}^\top \Zc  & 0  \\
\lambda Y & 0 \\
\mb{A}^{-1/2} \smat{P\\ Y\\ (I_m \otimes \bar{x})Y + (\bar{u} \otimes I_n)P\\ 0 } & 0 \\
\Lambda \mb{Q}^{1/2} & 0 
\end{matrix}
\right. \\
& 
\left.
\begin{matrix}
\dots & \star & \star  & \star & \star\\
\dots & \star & \star & \star & \star\\
\dots & -\lambda (I_m \otimes P) & \star & \star & \star\\
\dots & 0 & - \lambda I_m & \star & \star\\
\dots & \mb{A}^{-1/2} \smat{0\\ 0\\  I_m \otimes P \\ 0} & 0 & -\Lambda I_{n+m+nm+1} & \star\\
\dots & 0 & 0 & 0 & - \Lambda I_n\\
\end{matrix}
\right]
\preceq 0.
\end{align*}
By a change of decision variables, this condition and $P \succ 0$ hold if and only if the program in~\eqref{progr_K_bil_setpoint_CT} is feasible.
In summary, we have shown that \eqref{progr_K_bil_setpoint_CT} implies \eqref{attractivity_CT}, thereby concluding the proof.
\endgroup
\end{proof}

Let us comment Theorem~\ref{thm:UNknown_u_bar}.
Due to the impossibility to achieve exact setpoint regulation for all dynamics consistent with data, see Lemma~\ref{lemma:cannot_find_u_bar}, we resort to the relaxation in Lemma~\ref{lemma:design_u_bar_gamma} and, accordingly, Theorem~\ref{thm:UNknown_u_bar} ensures local asymptotic stability of the set $\{ x \in \real^n \colon (x- \bar{x})^\top P^{-1} (x- \bar{x}) \le \eta \}$, rather than of the point $\bar{x} = 0$ as in Theorem~\ref{thm:known_u_bar:dt}.
To make this set ``small'', a small value for $\eta$ is beneficial.
As for the design parameter $\epsilon$, a quite small value for $\epsilon$ can be used.
As in Theorem~\ref{thm:known_u_bar:dt}, we need a line search in $\lambda$, which is due to the products $\lambda Y$ and $\lambda (I_m \otimes P)$ between decision variables and arises from dealing with the bilinear dynamics; we need a second line search in $s$ due to the product $s P$, which arises from dealing with the disturbance-like term $\ABCd\smat{\bar{x}\\ \bar{u}\\ (I_m \otimes \bar{x}) \bar{u}\\ 1}$, cf.~\eqref{sys_ct_tilde_2}, whose size is minimized but not zero for all $\ABCd \in \mathcal{C}$.
Indeed, the same type of line search occurs in a model-based linear-time-invariant setting for unit peak inputs, see \cite[p.~83]{boyd1994linear}.
We note that instead of decision variables $(s,\tau_\eta,\tau_\gamma)$ satisfying
\begin{align*}
\tau_\eta \ge 0, \tau_\gamma \ge 0, 
s + \tau_\eta \ge 0, \tau_\eta (\eta-1) = \epsilon + s + \tau_\gamma \gamma,
\end{align*}
we can equivalently take decision variables $(s,\tau_\gamma)$ satisfying
\begin{align}
& \epsilon + s + \tau_\gamma \gamma \le 0, \tau_\gamma \ge 0, \epsilon + s \eta + \tau_\gamma \gamma \le 0 \notag \\
& \iff \tau_\gamma \ge 0, s \le -\tfrac{\epsilon + \tau_\gamma \gamma}{\eta} \label{bound_s}
\end{align}
thanks to $\eta \in (0,1)$ and $\epsilon > 0$.

A discrete-time counterpart of Theorem~\ref{thm:UNknown_u_bar} can be obtained along similar lines, but we omit it due to space limitations.

\section{Numerical validation}
\label{sec:example}

\begin{figure}
\centerline{\includegraphics[scale=0.85]{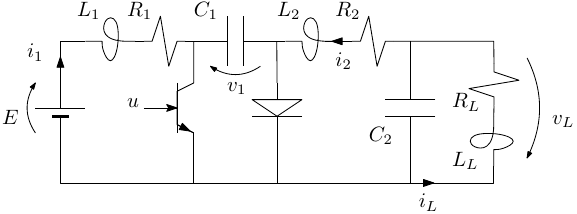}}
\caption{Electrical circuit of a \'Cuk converter as in~\cite{amato2009stabilization}.}
\label{fig:cuk}
\end{figure}

In this section we numerically illustrate our results.
Consider a \'Cuk converter as in~\cite{amato2009stabilization} and depicted in Figure~\ref{fig:cuk}.
After state-space averaging, a \'Cuk converter is described by
\begingroup
\setlength\arraycolsep{1.pt}
\begin{align}
& \dot{x} = A_\star x + B_\star u + C_\star (I_m \otimes x) u + d_\star \notag \\
& = \smat{
   -1.0000   & -1.0000  &         0   &       0  &        0\\
    0.0100   &       0  &         0   &       0  &        0\\
         0   &       0  &   -0.5000   &       0  &  -1.0000\\
         0   &       0  &         0   &-150.000  &  10.0000\\
         0   &       0  &    0.1000   & -0.1000  &        0\\                         
} x +
\smat{
0\\
0\\
0\\
0\\
0\\
} u \notag \\
& 
\quad +
\smat{
0        &   1.0000 &         0 &         0 & 0\\
-0.0100  &        0 &   -0.0100 &         0 & 0\\
0        &   1.0000 &         0 &         0 & 0\\
0        &        0 &         0 &         0 & 0\\
0        &        0 &         0 &         0 & 0\\
} 
(I_m \otimes x) u + 
\smat{    
	30\\
     0\\
     0\\
     0\\
     0} \label{Cuk_matrices}
\end{align}
\endgroup
where $x = (i_1, v_1, i_2, i_L, v_L)$ are the state variables, the duty cycle $u$ is the input, the parameters in $\ABCd[\star]$ correspond to the expressions in \cite[Eq.~(15)]{amato2009stabilization} after substituting the values in \cite[Table 1]{amato2009stabilization}.
The matrices $\ABCd[\star]$ are unknown and not used in our data-based designs.
We consider a desired setpoint $\bar{x} = (2.23, 58.76, 2.00, 2.00, 30.00)$ and $\bar{x}$ can be enforced by $\bar{u} = 0.527480$, which is assumed known in Section~\ref{sec:known_u_bar} and unknown in Section~\ref{sec:UNknown_u_bar}.

\begin{figure}
\centerline{\includegraphics[scale=0.62]{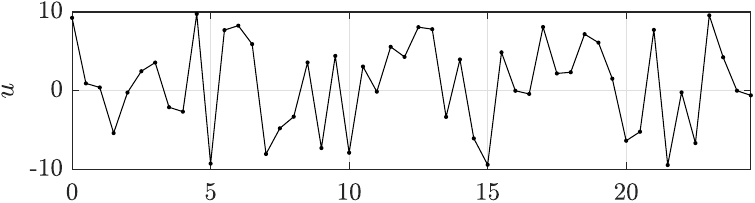}}
\centerline{\includegraphics[scale=0.62]{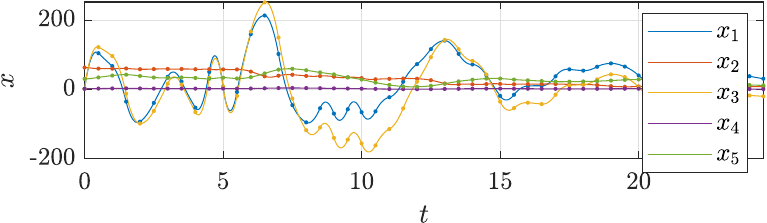}}
\caption{Input and state of the bilinear system during the experiment: dots correspond to samples.
}
\label{fig:data}
\end{figure}

We collect data in an open-loop experiment on the bilinear system~\eqref{Cuk_matrices}.
We use a random signal $u$ while the process noise $e$ is also acting on the bilinear system, as detailed in Section~\ref{sec:data}.
As for $e$, we take the bound $\Xi \Xi^\top = 10^{-4} I$ in the set $\mathcal{E}$ in~\eqref{set_E}.
The input and state signals, along with the $T=50$ samples we use, are depicted in Figure~\ref{fig:data}.
This same dataset is used to design control laws according to Theorem~\ref{thm:known_u_bar:ct} or Lemma~\ref{lemma:design_u_bar_gamma} plus Theorem~\ref{thm:UNknown_u_bar}, through MATLAB\textsuperscript{\textregistered}, MOSEK and YALMIP \cite{lofberg2004yalmip}.

We first consider the approach in Section~\ref{sec:known_u_bar} where $\bar{u}$ is known.
In~\eqref{sol-bil:ct}, the products between decision variables require a line search in $\lambda$ in the interval $[0,5]$, for 50 values of $\lambda$.
The values of $\lambda$ yielding a feasible solution in~\eqref{sol-bil:ct} are in Figure~\ref{fig:lambda_search}: corresponding to each of these $\lambda$, we depict the volume $\sqrt{\det(P)}$ and the diameter $2 \lambda_{\max}(P^{1/2})$ \cite[\S 3.7]{boyd1994linear} of the ellipsoids of the guaranteed basin of attraction $\mathcal{B}_P^{\bar{x}}$.
In general, both of them may be used as convex objective functions in addition to~\eqref{sol-bil:ct}.
Based on Figure~\ref{fig:lambda_search}, we select $\lambda = 0.102$, corresponding to the largest volume and diameter (for the chosen grid), to validate the control in~\eqref{control_law} in closed loop.
For $\lambda = 0.102$, \eqref{sol-bil:ct} returns
\begin{align*}
K & =
\smat{
-0.0747&   -0.0356&   -0.0752&   -9.7590&    0.6449&
},\\
P & = 10^3 \cdot
\smat{
    3.8336&   -0.4901&   -3.1771&    0.1096&    1.6418\\
   -0.4901&    0.7333&    0.1821&    0.0068&    0.1019\\
   -3.1771&    0.1821&    4.7965&   -0.1694&   -2.5369\\
    0.1096&    0.0068&   -0.1694&    0.0181&    0.2720\\
    1.6418&    0.1019&   -2.5369&    0.2720&    4.0786\\
}.
\end{align*}
A resulting solution, with initial condition in the guaranteed basin of attraction, is in Figure~\ref{fig:cl_known} and converges to $\bar{x}$.

\begin{figure}
\centerline{\includegraphics[scale=0.59]{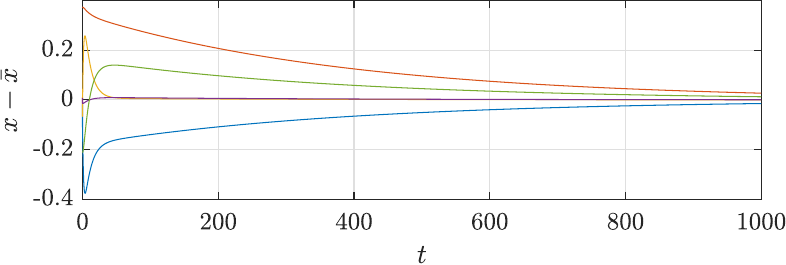}}
\caption{Closed-loop solution for the controller designed as in Section~\ref{sec:known_u_bar}.}
\label{fig:cl_known}
\end{figure}

We then consider the approach of Section~\ref{sec:UNknown_u_bar} where $\bar{u}$ is no longer known but a design parameter.
The first step in this case is solving~\eqref{min_progr_Gamma_bil_equiv_CT}, which yields $\gamma = 1.7251\cdot 10^{-5}$ and $\bar{u} = 0.527477$.
The fact that $\gamma$ is quite small and the designed $\bar{u}$ almost coincides with the actual equilibrium input indicates that, for this example and the uncertainty in $\mathcal{C}$, \eqref{min_progr_Gamma_bil_equiv_CT} or the equivalent \eqref{min_progr_Gamma_CT} are very good proxies for \eqref{ideal_Gamma_for_all_ABC_CT}.
With this $(\gamma,\bar{u})$, the second step is solving \eqref{progr_K_bil_setpoint_CT} for $\eta = 0.1$ and $\epsilon = 10^{-3}$.
In~\eqref{progr_K_bil_setpoint_CT}, the products between decision variables require line searches in both variables $\lambda$ and $s$.
We consider 10 values of $\lambda \in [0.6,1.5]$ and 20 values of $s \in [-0.05, -\frac{\epsilon}{\eta}]$, where the intervals reflect the constraints $\lambda >0$ and $s \le -\tfrac{\epsilon + \tau_\gamma \gamma}{\eta} \le -\tfrac{\epsilon}{\eta}$ from~\eqref{bound_s}.
The values of $(\lambda,s)$ yielding a feasible solution in~\eqref{progr_K_bil_setpoint_CT} are in Figure~\ref{fig:lambda_s_search}, where we depict volume and diameter of the ellipsoid of the guaranteed basin of attraction.
Based on Figure~\ref{fig:lambda_s_search}, we select $(\lambda,s) = (1, -0.03316)$, for which \eqref{progr_K_bil_setpoint_CT} returns
\begin{align*}
K & =
\smat{
   -0.1577&   -2.7357&   -0.1566&   14.3106&   -1.1556
},\\
P & =
\smat{
  133.0138&  -13.1926&  -46.5750&    6.9404&  104.0387\\
  -13.1926&    1.6424&    2.3794&   -0.8529&  -12.7918\\
  -46.5750&    2.3794&   57.2233&   -2.4462&  -36.6814\\
    6.9404&   -0.8529&   -2.4462&    0.5014&    7.5173\\
  104.0387&  -12.7918&  -36.6814&    7.5173&  112.7307\\
}.
\end{align*}
A resulting solution, with initial condition in the guaranteed basin of attraction, is in Figure~\ref{fig:cl_unknown} and converges to a (very small) neighborhood of $\bar{x}$.

\begin{figure}
\centerline{\includegraphics[scale=0.6]{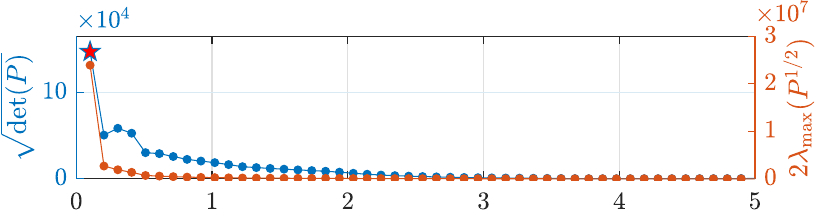}}
\caption{Volume $\sqrt{\det P}$ and diameter $2\lambda_{\max}(P^{1/2})$ of the guaranteed basin of attraction for $\lambda$'s yielding a feasible solution in~\eqref{sol-bil:ct}.}
\label{fig:lambda_search}
\end{figure}

\begin{figure}
\centerline{\includegraphics[scale=0.59]{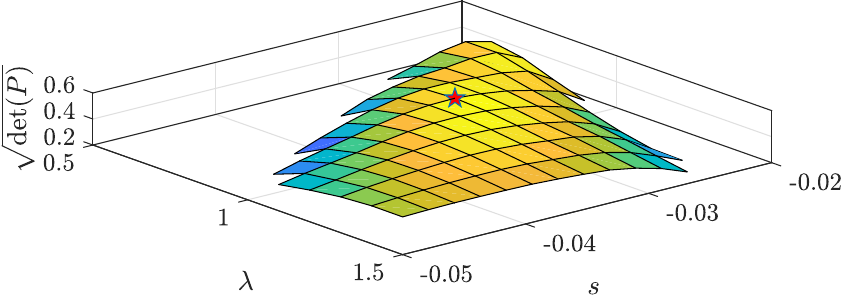}}
\centerline{\includegraphics[scale=0.59]{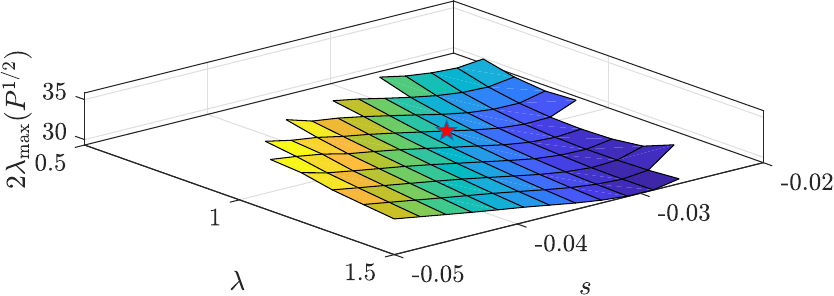}}
\caption{Volume $\sqrt{\det P}$ (top) and diameter $2\lambda_{\max}(P^{1/2})$ (bottom) of the guaranteed basin of attraction for $\lambda$'s yielding a feasible solution in~\eqref{progr_K_bil_setpoint_CT}.}
\label{fig:lambda_s_search}
\end{figure} 

\begin{figure}
\centerline{\includegraphics[scale=0.59]{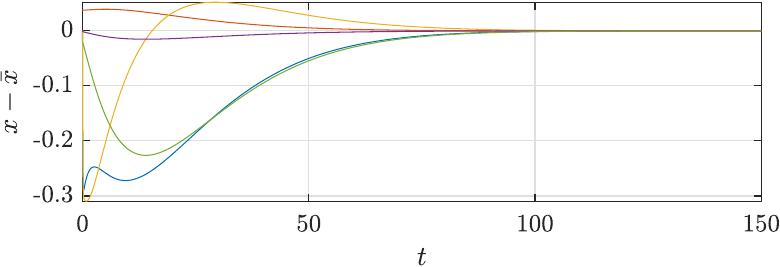}}
\caption{Closed-loop solution for the controller designed as in Section~\ref{sec:UNknown_u_bar}.}
\label{fig:cl_unknown}
\end{figure}

\section{Conclusion}

Starting from noisy data collected from a bilinear system, we have proposed two designs that aim at asymptotically stabilizing a given state setpoint with a guaranteed basin of attraction and take the form of linear matrix inequalities (modulo line searches on scalar variables).
These two designs are for when the equilibrium input corresponding to the setpoint is known or unknown.
An interesting conclusion for the second case is that, in the considered data-based setting, one can only achieve stabilization of a ``small'' neighborhood of the given setpoint.

\bibliographystyle{plain}
\bibliography{pubs-bil-petersen.bib}

\end{document}